\documentclass[11pt]{amsart}
\usepackage{amsmath}

\usepackage[T1]{fontenc}
\usepackage[latin9]{inputenc}
\usepackage{geometry}
\usepackage{epsfig}
\usepackage{verbatim}
\usepackage{color}
\usepackage{graphicx,amssymb,amsmath,amsthm}
\usepackage[colorlinks, linkcolor=red, anchorcolor=blue, citecolor=green]{hyperref}
\usepackage{mathrsfs}
\usepackage{hyperref} 
\usepackage{amssymb}
\usepackage{enumerate}
\usepackage{pifont}
\usepackage{chngcntr}

\renewcommand{\eqref}[1]{(\ref{#1})}

\newcommand{\vA}{{\boldsymbol A}}

\newcommand{\vx}{{\boldsymbol x}}
\newcommand{\vy}{{\boldsymbol y}}
\newcommand{\vz}{{\boldsymbol z}}

\newcommand{\va}{{\boldsymbol a}}
\newcommand{\vb}{{\boldsymbol b}}

\renewcommand{\top}{T}
\newtheorem{prop}{Proposition}[section]

\newtheorem{lemma}[prop]{Lemma}
\newtheorem{defi}{Definition}[section]

\newtheorem{theorem}[prop]{Theorem}

\newtheorem{remark}[prop]{Remark}

\usepackage{amsmath}
\allowdisplaybreaks[4]

\date{}

\begin{document}
\bibliographystyle{plain}
\title{The Optimal Condition Number for ReLU Function}

\author{Yu Xia}
\thanks{ Yu Xia was supported by NSFC grant (12271133, U21A20426, 11901143) and the key project of Zhejiang Provincial Natural Science Foundation grant (LZ23A010002)}
\address{School of Mathematics, Hangzhou Normal University, Hangzhou 311121, China.}
\email{yxia@hznu.edu.cn}

\author{Haoyu Zhou}
\address{School of Mathematics, Hangzhou Normal University,
Hangzhou, 311121, China.}
\email{2020210222016@stu.hznu.edu.cn}

\maketitle

\begin{abstract}
ReLU is a widely used activation function in deep neural networks. This paper explores the stability properties of the ReLU map. For any weight matrix $\boldsymbol{A} \in \mathbb{R}^{m \times n}$ and bias vector $\vb \in \mathbb{R}^{m}$ at a given layer, we define the condition number $\beta_{\vA,\vb}$ as $\beta_{\vA,\vb} = \frac{\mathcal{U}_{\vA,\vb}}{\mathcal{L}_{\vA,\vb}}$, where $\mathcal{U}_{\vA,\vb}$
  and $\mathcal{L}_{\vA,\vb}$ are the upper and lower Lipschitz constants, respectively. We first demonstrate that for any given $\boldsymbol{A}$ and $\vb$, the condition number satisfies $\beta_{\vA,\vb} \geq \sqrt{2}$. Moreover, when the weights of the network at a given layer are initialized as random i.i.d. Gaussian variables and the bias term is set to zero, the condition number asymptotically approaches this lower bound. This theoretical finding suggests that Gaussian weight initialization is optimal for preserving distances in the context of random deep neural network weights.

\end{abstract}
\vspace{0.4cm}


	\section{Introduction}
	\subsection{Problem Setup}\label{sec: intro}
	Deep Neural Networks have achieved remarkable success in various domains, including signal processing, computer vision, and natural language processing \cite{DNN3, DNN2, DNN1}. Unlike traditional activation functions such as the sigmoid function and the hyperbolic tangent function, the Rectified Linear Unit (ReLU) activation function is defined as follows:
	\[
	\sigma(x): = \max(0, x),
	\]
	which effectively mitigates the gradient vanishing problem in deep architectures.  However, the successive application of linear operators and ReLU functions can potentially distort distance metrics across the layers of the network. This raises a fundamental question:  
	 {\em  To what extent are the metric properties preserved across different layers of the network?}

Specifically, for the $i$-th layer, we denote the transformation as $T_i: \mathbb{R}^{n_i} \rightarrow \mathbb{R}^{n_{i+1}}$. Let $c_i$  and $C_i$ be non-negative constants such that, for any two points $\vx, \vy \in \mathbb{R}^{n_i}$,
\[
c_i\cdot \|\vx - \vy\|_2 \leq \|T_i(\vx) - T_i(\vy)\|_2 \leq C_i\cdot \|\vx - \vy\|_2.
\]
The closer the ratio $\frac{C_i}{c_i}$ is to 1, the more stable the $i$-th layer is.

The bi-Lipschitz constants under the ReLU map with matrix $\vA\in \mathbb{R}^{m\times n}$ and vector $\vb \in \mathbb{R}^{n}$ are characterized by two non-negative constants, $\mathcal{L}$ and $\mathcal{U}$. These constants are defined such that for any $\boldsymbol{x},\boldsymbol{y}\in \mathbb{R}^n$, the following inequality holds:
\begin{equation}
\label{original_lipschitz1}
\mathcal{L}\cdot \|\boldsymbol{x}-\boldsymbol{y}\|_2\leq \frac{1}{\sqrt{m}}\cdot \|\sigma(\boldsymbol{A}\boldsymbol{x}+\vb)-\sigma(\boldsymbol{A}\boldsymbol{y}+\vb)\|_2\leq \mathcal{U}\cdot \|\boldsymbol{x}-\boldsymbol{y}\|_2.
\end{equation}
Here $\boldsymbol{A}\in \mathbb{R}^{m\times n}$ and $\vb\in \mathbb{R}^{n}$ can be considered as  the weight matrix and the bias term, respectively, at a given layer of the network. The dimensions $m$ and $n$ correspond to the input and output dimensions of the corresponding layer. The factor $\frac{1}{\sqrt{m}}$ serves to normalize the output.
We denote the optimal (i.e., greatest possible) lower bound $\mathcal{L}$ as $\mathcal{L}_{\vA,\vb}$ and the optimal (i.e., smallest possible) upper bound $\mathcal{U}$ as $\mathcal{U}_{\vA,\vb}$. These optimal constants are formally defined as:
\begin{equation}\label{lip_constant1}
\mathcal{L}_{\vA,\vb}:=\inf_{\boldsymbol{x},\boldsymbol{y}\in \mathbb{R}^{n}, \boldsymbol{x}\neq \boldsymbol{y}}\frac{1}{\sqrt{m}}\cdot \frac{\|\sigma(\boldsymbol{A}\boldsymbol{x}+\vb)-\sigma(\boldsymbol{A}\boldsymbol{y}+\vb)\|_2}{\|\boldsymbol{x}-\boldsymbol{y}\|_2},
\end{equation}
and 
\begin{equation}\label{lip_constant2}
  \mathcal{U}_{\vA,\vb}:=\sup_{\boldsymbol{x},\boldsymbol{y}\in \mathbb{R}^{n}, \boldsymbol{x}\neq \boldsymbol{y}}\frac{1}{\sqrt{m}}\cdot \frac{\|\sigma(\boldsymbol{A}\boldsymbol{x}+\vb)-\sigma(\boldsymbol{A}\boldsymbol{y}+\vb)\|_2}{\|\boldsymbol{x}-\boldsymbol{y}\|_2}.
\end{equation}

In this paper, we focus on investigating the condition number $\beta_{\boldsymbol{A},\boldsymbol{b}}$, which is formally defined as:
\begin{equation}\label{beta_def}
\beta_{\vA,\vb}:=\frac{\mathcal{U}_{\vA,\vb}}{\mathcal{L}_{\vA,\vb}}.
\end{equation}

It serves as a quantitative measure for assessing the stability and isometric properties of specific neural network layers. Notably, $\beta_{\vA,\vb}$ is invariant under positive scaling, meaning that $\beta_{\vA,\vb} = \beta_{c\vA, c\vb}$ for any positive constant $c \in \mathbb{R}$.

In this paper, we address two  questions:
 \begin{enumerate}[(i)]
  \item For any fixed weight matrix $\vA \in \mathbb{R}^{m \times n}$ and bias term $\vb \in \mathbb{R}^{m}$, what is the lower bound of the condition number $\beta_{\vA,\vb}$? 
  \item Do there exist specific $\vA$ and $\vb$ (particularly random generalizations) that achieve the optimal (i.e. smallest) condition number bound? 
  \end{enumerate}

\subsection{Related Works}
Several historical works have examined the bi-Lipschitz constants $\mathcal{U}_{\vA,\vb}$  and $\mathcal{L}_{\vA,\vb}$. {Since the ReLU function is a non-expansive mapping, it follows that for any fixed $\vA$ and $\vb$, we have 
\[
\mathcal{U}_{\vA,\vb} \leq \lambda_{\max}(\vA),
\]
 where $\lambda_{\max}(\vA)$ denotes the largest singular value  of $\vA$ \cite{lip_upper}. 
 
 Determining $\mathcal{L}_{\vA,\vb}$
  is more challenging. When setting the bias term as $\vb = \boldsymbol{0}$, it has been demonstrated in \cite[Theorem 3]{low_lip} that, if the function $f:\mathbb{R}^{n}\rightarrow \mathbb{R}^{m}$ such that $f(\vx)=\sigma(\vA\vx)$ is injective, then 
  \begin{equation}
  \label{eqn: L_a_temp}
  \mathcal{L}_{\vA,\boldsymbol{0}} \geq \sqrt{\frac{\lambda(\vA)}{2m}},
  \end{equation}
 where 
 \[
 \lambda(\vA):=\min_{\vx\in \mathbb{R}^{n}}\lambda_{\min}(\vA_{S(\vx,\vA)}).
 \]
  Here, $\lambda_{\min}(\cdot)$ denotes the smallest singular value, and $\vA_{S(\vx,\vA)}$ represents the submatrix of $\vA$ comprising the row vectors of $\vA$ indexed by the set 
  \[
  S(\vx,\vA):=\{j\in \mathbb{Z}\ :\ \langle \va_{j},\vx\rangle\geq 0\ \text{and}\ 1\leq j\leq m\},
  \] where $\va_{j}\in \mathbb{R}^{n}$ is the $j$-th row vector of $\vA$. 
    In \cite[Theorem 1]{Freeman}, the lower bound in (\ref{eqn: L_a_temp}) was refined to 
  \[
  \frac{1}{2}\sqrt{\lambda(\vA)} \leq \mathcal{L}_{\vA,\boldsymbol{0}} \leq \sqrt{\lambda(\vA)}.
  \]
   Consequently, the upper bound on the condition number $\beta_{\vA,\boldsymbol{0}}$ can be expressed as 
   \[
   \beta_{\vA,\boldsymbol{0}} \leq 2\sqrt{\frac{\lambda_{\max}(\vA)}{\lambda(\vA)}}.
   \]
    However, the preceding discussion does not provide a lower bound for $\beta_{\vA,\boldsymbol{0}}$, aside from the trivial lower bound $\beta_{\vA,\vb} \geq 1$.}

Furthermore, random initialization of weights is often preferred over sampling parameters from pre-trained networks \cite{random1, lip3, random2,lip2,DNN1}. Research indicates that the reconstruction error of randomly initialized neural networks yields comparable F1 scores \cite{Random}. However, unlike linear random projections \cite{CS1, CS2}, the theoretical analysis of the distance-preserving properties in neural networks becomes more complicated due to the inherent nonlinearities of the ReLU function \cite{lip3,candes}.

In the context of random weights, \cite[Theorem 2]{low_lip} establishes the injectivity for Gaussian weights. Let $\vA\in \mathbb{R}^{m\times n}$ be a standard Gaussian random matrix, and define $\mathcal{I}(m,n):=\mathbb{P}\{\vx\mapsto \text{ReLU}(\vA\vx)\text{ is injective}\}$. Then, as $n\rightarrow \infty$,
\[
\lim_{n\rightarrow \infty}\mathcal{I}(m,n)=1,\quad \text{if}\ \frac{m}{n}\geq 10.5\qquad \text{and}\qquad \lim_{n\rightarrow \infty}\mathcal{I}(m,n)=0,\quad\text{if}\ \frac{m}{n}\leq 3.3.
\] 
Nevertheless, the theoretical results do not yield precise bounds on the bi-Lipschitz constants for the ReLU mapping under Gaussian random weight matrices. 

The studies conducted in \cite{ref1,ref2}  demonstrate that for $\boldsymbol{A}\in \mathbb{R}^{m\times n}$ with elements 
  drawn from a standard Gaussian distribution, the following property holds with high probability for any fixed $\delta > 0$:
\begin{equation}\label{modified_lipschitz}
\frac{1}{4}\|\boldsymbol{x}-\boldsymbol{y}\|_2^2-\delta \leq \|\sigma(\boldsymbol{A}\boldsymbol{x})-\sigma(\boldsymbol{A}\boldsymbol{y})\|_2^2\leq \frac{1}{2}\|\boldsymbol{x}-\boldsymbol{y}\|_2^2+\delta,
\end{equation}
for all vectors $\boldsymbol{x}, \boldsymbol{y} \in \mathbb{R}^n$, provided that $m \geq C_{\delta} \cdot n$, where $C_{\delta}$ is a positive constant dependent on $\delta$. 

Numerous studies have indicated that under the conditions described in (\ref{modified_lipschitz}), distance can be preserved in a network with respect to a specific layer, as noted in \cite{lip3,wrong_lip}. However, it is important to recognize that the result in (\ref{modified_lipschitz}) does not fully capture the bi-Lipschitz property for positive and finite $\mathcal{L}$ and $\mathcal{U}$ as described in (\ref{original_lipschitz1}), where $\vb=\boldsymbol{0}$. To illustrate this limitation, consider the case where $\boldsymbol{x} \neq \boldsymbol{y}$. In this scenario, (\ref{modified_lipschitz}) leads to the following formulation:
\begin{equation}\label{wrong_lipschitz}
\frac{\frac{1}{4}\|\boldsymbol{x}-\boldsymbol{y}\|_2^2-\delta}{\|\boldsymbol{x}-\boldsymbol{y}\|_2^2} \leq \frac{\|\sigma(\boldsymbol{A}\boldsymbol{x})-\sigma(\boldsymbol{A}\boldsymbol{y})\|_2^2}{\|\boldsymbol{x}-\boldsymbol{y}\|_2^2}\leq \frac{\frac{1}{2}\|\boldsymbol{x}-\boldsymbol{y}\|_2^2+\delta}{\|\boldsymbol{x}-\boldsymbol{y}\|_2^2}.
\end{equation}
This formulation reveals a critical issue: when the parameter $\delta$, the input dimension $n$, and the output dimension $m$ are specified, as $\boldsymbol{y}$ approaches $\boldsymbol{x}$, the lower and upper bounds of (\ref{wrong_lipschitz}) diverge to negative and positive infinity, respectively. Given that positive and finite estimates for $\mathcal{L}$ and $\mathcal{U}$ cannot be derived, it follows that a bound on the condition number $\beta_{\vA,\vb}$ based on (\ref{modified_lipschitz}) is unattainable. Therefore, we must employ more delicate tools to refine the results presented in (\ref{modified_lipschitz}).

\subsection{Our Contributions}

Our investigation addresses the questions posed in Section \ref{sec: intro}.

Regarding Question (i), we demonstrate in Theorem \ref{thm: general_one} that for any arbitrary matrix $\boldsymbol{A} \in \mathbb{R}^{m \times n}$  and vector $\vb \in \mathbb{R}^{m}$, the condition number $\beta_{\vA,\vb}$ is bounded below by $\sqrt{2}$. Formally, we have:
\begin{equation}
\label{eqn: result1}
\beta_{\vA,\vb}\geq \beta_{\vA,\boldsymbol{0}}\geq\sqrt{2},
\end{equation}
   for any $\vb \in \mathbb{R}^{m}$.
This proven lower bound of $\sqrt{2}$
  for the condition number, provides a nontrivial lower bound that answers Question (i). It establishes a fundamental limit on the stability of any single layer in a ReLU network and serves as a benchmark against which various weight initialization strategies can be evaluated.

Turning to Question (ii), Theorem \ref{thm: Gaussian_bilip} examines the scenario where $\boldsymbol{A} \in \mathbb{R}^{m \times n}$
  is a standard Gaussian random matrix, aligning with a variant of the He initialization \cite{He}. Based on the result in (\ref{eqn: result1}), we only consider the case when $\vb=\boldsymbol{0}$. Our analysis reveals that for any fixed $\delta \in (0, 4/5)$, with high probability, the following holds:
 \[
\beta_{\vA,\boldsymbol{0}}\leq \sqrt{2}+\delta,
 \]
provided that $m \gtrsim C_{\delta} \cdot n$, where $C_{\delta}$ is a constant dependent on a positive parameter. Consequently, for a standard Gaussian random matrix, $\beta_{\vA,\boldsymbol{0}}$ asymptotically converges to the lower bound $\sqrt{2}$ as $m$ approaches infinity. This result formally establishes the distance-preserving property for Gaussian random weights in the context of ReLU neural networks, suggesting that random weights are not merely a convenient starting point, but may indeed be an optimal choice for maintaining stable signal propagation through the network layers.

In addition to our main theorems, we have obtained several byproducts of independent interest in Section \ref{sec: gaussian_byproduct}. To derive the results in Theorem \ref{thm: Gaussian_bilip}, we established a more general result applicable to any arbitrary cone $S \subseteq \mathbb{R}^{n}$, as presented in Theorem \ref{thm: lipschitz_result}. While Theorem  \ref{thm: lipschitz_result} shares certain similarities with the result presented in \cite[Theorem 3]{ref1}, our work significantly enhances the sampling complexity. Specifically, we have reduced the dependency from $\omega^4(\cdot)$ to $\omega^2(\cdot)$, where $\omega(\cdot)$ denotes the Gaussian width as defined in Definition \ref{eqn: gaussian_width}. More details can be seen in Remark \ref{rem: key_rem}. 

Besides,  Theorem  \ref{thm: lipschitz_result} builds upon the findings in Theorem \ref{thm: large} and Theorem \ref{thm: small}, which separately address cases where the distance between $\vx \in \mathbb{R}^{n}$ and $\vy \in \mathbb{R}^{n}$ is small or large, respectively. Furthermore, Theorem \ref{thm: small}, which focuses on scenarios where the distance between $\vx$ and $\vy$ is small, shows that the output distance $\frac{1}{\sqrt{m}}\|\sigma(\vA\vx)-\sigma(\vA\vy)\|^2_2$ is around $\frac{1}{2}\|\vx-\vy\|_2^2$. It also merits independent consideration, as the techniques employed in this theorem differ substantially from those in Theorem \ref{thm: large}, playing a crucial role in achieving effective condition number estimation.

These theoretical results thus bridge the gap between empirical practices in deep learning and their mathematical foundations, providing a solid theoretical justification for the use of random weights in neural network initialization.
\subsection{Notations and Definitions}
	For convenience, take $\mathbb{S}^{n-1}$ and $\mathbb{B}^{n}$ as unit sphere and unit ball in $n$-dimensional space:
		 \begin{equation}\label{eqn: S_n and B_n}
		\mathbb{S}^{n-1}:=\{\boldsymbol{z}\in \mathbb{R}^n\ :\ \|\boldsymbol{z}\|_2= 1\},\quad \text{and}\quad \mathbb{B}^n:=\{\boldsymbol{z}\in \mathbb{R}^n\ :\ \|\boldsymbol{z}\|_2\leq 1\}.
		 \end{equation}
		For a specific set $S\subseteq \mathbb{R}^{n}$, we define the difference set $S-S$ as: 
	\begin{equation}
\label{eqn: S-S}
S-S:=\{\vx-\vy\ :\ \vx\in S\ \text{and}\ \vy\in S\}. 
\end{equation}	
Additionally, we denote the function $\phi(\vx,\vy):\mathbb{R}^{n}\times \mathbb{R}^{n}\rightarrow \mathbb{R}$ as: 
		\begin{equation}
		\label{eqn: phi}
		\phi(\vx,\vy):=\frac{\left(\sin \theta_{\vx,\vy}-\theta_{\vx,\vy}\cdot \cos \theta_{\vx,\vy}\right)}{\pi}\cdot \frac{\|\vx\|_2\cdot \|\vy\|_2}{\|\vx-\vy\|_2^2},
		\end{equation}
		where 
		\begin{equation}\label{eqn: theta}
		\theta_{\vx,\vy}:=\cos^{-1}\left(\frac{\langle \vx,\vy\rangle}{\|\vx\|_2\cdot \|\vy\|_2}\right)\in [0,\pi].
		\end{equation}
		Here we adopt the convention that  $\phi(\vx,\vy):=0$ if $\vx=\boldsymbol{0}$ or $\vy=\boldsymbol{0}$. 
	
		It is worth noting that throughout the subsequent analysis, we employ the notation $c_A$  and $C_A$ to represent constants dependent on the parameter $A$. These constants should be understood as potentially distinct across various theorems. Moreover, we employ the following notational conventions: $A\lesssim B$ denotes that there exists an absolute positive constant $C$ such that $A\leq CB$. Similarly, $A\gtrsim B$ indicates that there exists an absolute positive constant $C$ such that $A\geq CB$.
		
\subsection{Orgainizations}
The structure of this paper is as follows: Section \ref{sec: main_result} presents the main theorems that address the questions raised in Section \ref{sec: intro}. In Theorem \ref{thm: general_one}, we establish the lower bound of the condition number $\beta_{\vA,\vb}$  for any $\vA\in \mathbb{R}^{m\times n}$ and $\vb\in \mathbb{R}^{m}$. In Theorem \ref{thm: Gaussian_bilip}, we demonstrate that the lower bound can be asymptotically achieved by selecting $\vA$ as a Gaussian weight matrix and setting $\vb=\boldsymbol{0}$. A more general result is presented in Theorem \ref{thm: lipschitz_result} in Section \ref{sec: gaussian_byproduct}, which can also be considered as a byproduct of independent interest. In this theorem, we establish the bi-Lipschitz constant for Gaussian random matrices for any arbitrary cone $S\subseteq\mathbb{R}^{n}$. This result heavily relies on Theorem \ref{thm: large} and Theorem \ref{thm: small} in Section \ref{sec: real_lipschitz}, which address the cases of large and small distances between $\vx,\vy\in S$, respectively. The proofs of Theorem \ref{thm: large} and Theorem \ref{thm: small} can be found in Section \ref{prof_large_sec} and Section \ref{proof of small}, respectively.

\section{Main Results}\label{sec: main_result}
We begin by establishing a universal lower bound for the condition number $\beta_{\boldsymbol{A},\vb}$ associated with arbitrary matrice $\boldsymbol{A}\in \mathbb{R}^{m\times n}$ and vector $\vb\in \mathbb{R}^{m}$. This result serves as a benchmark for subsequent analyses of random matrices.
\begin{theorem}\label{thm: general_one} 
For any matrix $\boldsymbol{A}\in \mathbb{R}^{m\times n}$ and vector $\vb\in \mathbb{R}^{m}$, the condition number $ \beta_{\boldsymbol{A},\vb}$, as formulated in (\ref{beta_def}), satisfies the following inequalities:
\begin{equation}
 \beta_{\vA,\vb}\geq \beta_{\vA,\boldsymbol{0}}\geq \sqrt{2}. 
 \end{equation} 
 \end{theorem}
  \begin{proof}
  For any $\vx,\vy\in \mathbb{R}^{n}$ and $\vb\in \mathbb{R}^{m}$, through direct computation, we obtain:
	  \[
		\begin{split}
		&\lim_{\alpha\rightarrow +\infty}\frac{1}{\sqrt{m}}\cdot \frac{\|\sigma(\vA(\alpha \vx)+\vb)-\sigma(\vA(\alpha\vy)+\vb)\|_2}{\|\alpha\vx-\alpha\vy\|_2}\\
		=&\lim_{\alpha\rightarrow +\infty}\frac{1}{\sqrt{m}}\cdot \frac{\|\sigma(\vA \vx+\frac{1}{\alpha}\vb)-\sigma(\vA\vy+\frac{1}{\alpha}\vb)\|_2}{\|\vx-\vy\|_2}\\
		=&\frac{1}{\sqrt{m}}\cdot \frac{\|\sigma(\vA \vx)-\sigma(\vA\vy)\|_2}{\|\vx-\vy\|_2}.
		\end{split}
		\]
		Consequently, we can deduce:
		\[
		\mathcal{U}_{\vA,\vb}\geq \mathcal{U}_{\vA,\boldsymbol{0}}\quad \text{and}\quad \mathcal{L}_{\vA,\vb}\leq \mathcal{L}_{\vA,\boldsymbol{0}},
		\]
		which implies:
		\begin{equation}
		\label{eqn: v1}
		\beta_{\vA,\vb}\geq \beta_{\vA,\boldsymbol{0}}. 
		\end{equation}

Moreover, take $\vx_1,\vx_2\in \mathbb{R}^{n}$ such that 
		\[
		\vx_1:=\underset{\vx\in \mathbb{S}^{n-1}}{\text{argmax}}\ \frac{1}{\sqrt{m}}\cdot \frac{\|\sigma(\vA\vx)\|_2}{\|\vx\|_2}\quad \text{and}\quad \vx_2:=\underset{\vx\in \mathbb{S}^{n-1}}{\text{argmin}}\ \frac{1}{\sqrt{m}}\cdot \frac{\|\sigma(\vA\vx)\|_2}{\|\vx\|_2}.
		\]
			Through straightforward calculations, we can establish that  
			\begin{equation}
			\label{eqn: temp_lower1}
			\mathcal{U}_{\vA,\boldsymbol{0}}\geq \frac{1}{\sqrt{m}}\cdot \frac{\|\sigma(\vA\vx_1)\|_2}{\|\vx_1\|_2}.
			\end{equation}
On the other hand,  we derive the following inequality:
\begin{equation}
\label{eqn: temp_lower2}
\begin{aligned}
(\mathcal{L}_{\vA,\boldsymbol{0}})^2\leq &\frac{1}{m}\cdot \frac{\|\sigma(\vA\vx_2)-\sigma(\vA(-\vx_2))\|_2^2}{\|\vx_2-(-\vx_2)\|_2^2}=\frac{1}{m}\cdot \frac{\|\sigma(\vA\vx_2)\|_2^2+\|\sigma(-\vA\vx_2)\|_2^2}{4\|\vx_2\|_2^2}\\
=&\frac{1}{m}\cdot \frac{\|\sigma(\vA\vx_2)\|_2^2}{4\|\vx_2\|_2^2}+\frac{1}{m}\cdot \frac{\|\sigma(\vA(-\vx_2))\|_2^2}{4\|-\vx_2\|_2^2}\\
\leq & \frac{1}{2m}\cdot \frac{\|\sigma(\vA\vx_1)\|_2^2}{\|\vx_1\|_2^2}.
\end{aligned}
\end{equation}
The second equality  is a direct consequence of the orthogonality property:
\[
\langle \sigma(\boldsymbol{A}\boldsymbol{x}_2), \sigma(-\boldsymbol{A}\boldsymbol{x}_2) \rangle=0.
\]	 
Based on (\ref{eqn: temp_lower1}) and (\ref{eqn: temp_lower2}), we can directly deduce that
		\begin{equation}
		\label{eqn: v2}
		\beta_{\vA,\boldsymbol{0}}=\frac{\mathcal{U}_{\vA,\boldsymbol{0}}}{\mathcal{L}_{\vA,\boldsymbol{0}}}\geq \sqrt{2}.
		\end{equation}
	Thus (\ref{eqn: v1}) and (\ref{eqn: v2}) complete the proof.

   \end{proof}
		
Given that $\beta_{\vA,\boldsymbol{0}} \leq \beta_{\vA,\vb}$  holds for any $\vb \in \mathbb{R}^{m}$, we now focus our attention on analyzing the condition number of standard Gaussian random matrices $\boldsymbol{A} \in \mathbb{R}^{m\times n}$  in conjunction with $\vb = \boldsymbol{0}$. Theorem \ref{thm: Gaussian_bilip} that follows is a direct consequence of Theorem \ref{thm: lipschitz_result}, obtained by setting $S=\mathbb{R}^{n}$. Given that Theorem \ref{thm: lipschitz_result} relies on geometric concepts such as Gaussian width, we have placed its detailed discussion in Section \ref{sec: gaussian_byproduct} for those interested in a deeper exploration of these mathematical tools.

In this section, however, we concentrate on the primary conclusion concerning the condition number of Gaussian random weights under the ReLU mapping. This result provides insights into the properties of randomly initialized neural networks, bridging the gap between theoretical analysis and practical applications.
 
		\begin{theorem}\label{thm: Gaussian_bilip}
		Assume that  $\boldsymbol{A}\in \mathbb{R}^{m\times n}$ is a standard Gaussian random matrix, for any fixed positive constant $\delta<4/5$, with probability at least $1-2\exp(-c_{\delta}\cdot m)$, it obtains that  
		\begin{equation}
		\label{eqn: beta_upper}
		\beta_{\vA,\boldsymbol{0}}\leq \sqrt{2}+\delta,
		\end{equation}
		provided that $m\geq C_{\delta} \cdot n$. Here $c_{\delta}$ and $C_{\delta}$ are positive constants depending on $\delta$. 
		\end{theorem}
			\begin{proof}
		The proof is presented in Section \ref{sec: gaussian_byproduct}. 
		\end{proof}
		Our investigation indicates that, when $\vA$ is a Gaussian random matrix, $\beta_{\vA,\boldsymbol{0}}$  tends towards the optimal lower bound $\sqrt{2}$  as $m$ turns to infinity. This finding confirms the distance-preserving characteristic of Gaussian random weights in ReLU neural networks.

\section{Bi-Lipschitz Constants for Gaussian Random Matrices and Proof of Theorem \ref{thm: Gaussian_bilip}}\label{sec: gaussian_byproduct}

In this section, we extend our analysis beyond the conventional setting of $\mathbb{R}^n$, examining the behavior of the bi-Lipschitz constants on an arbitrary given cone $S \subseteq \mathbb{R}^n$, or equivalently, on sets satisfying $S=\lambda S$ for any positive $\lambda\in \mathbb{R}$. Specifically, we investigate constants $\mathcal{L}$ and $\mathcal{U}$ that satisfy the following inequality for any $\boldsymbol{x},\boldsymbol{y}\in S$:
\begin{equation}
\label{original_lipschitz2}
\mathcal{L}\cdot \|\boldsymbol{x}-\boldsymbol{y}\|_2\leq \frac{1}{\sqrt{m}}\cdot\|\sigma(\boldsymbol{A}\boldsymbol{x})-\sigma(\boldsymbol{A}\boldsymbol{y})\|_2\leq \mathcal{U}\cdot \|\boldsymbol{x}-\boldsymbol{y}\|_2,
\end{equation}
where $\vA\in \mathbb{R}^{m\times n}$ is a standard random Gaussian matrix. To establish estimates for $\mathcal{L}$ and $\mathcal{U}$, we first introduce two geometric measures:
		
		\begin{defi}[Gaussian width]\label{eqn: gaussian_width} 
		For any set $T\subset \mathbb{R}^{n}$, its Gaussian width is given by
\[
\omega(T) := \mathbb{E}\sup_{\boldsymbol{x}\in T}\langle \boldsymbol{g},\boldsymbol{x}\rangle,
\]
where $\boldsymbol{g}\sim \mathcal{N}(\boldsymbol{0},\boldsymbol{I}_n)$ is an $n$-dimensional standard Gaussian random vector.
		\end{defi}
		\begin{defi}[$\epsilon$-Net] \label{def:epsilon-net} Let $T$ be a subset of $\mathbb{R}^n$ and $\epsilon > 0$. A set $\mathcal{N}(T,\epsilon) \subset T$ is called an $\epsilon$-net of $T$: if for every $\mathbf{x} \in T$, there exists a point $\mathbf{x}_\epsilon \in \mathcal{N}(T,\epsilon)$ such that
$\|\mathbf{x} - \mathbf{x}_\epsilon\|_2 \leq \epsilon.$
The cardinality of $\mathcal{N}(T,\epsilon)$ is denoted by $\#\mathcal{N}(T,\epsilon)$. 
\end{defi}
The Gaussian width is a useful measure for the dimensionality of a set through the application of $\epsilon$-net cardinality bounds. Specifically, Dudley's and Sudakov's inequalities (Theorem 7.4.1 and Theorem 8.1.3 in \cite{vershynin}) establish a relationship between Gaussian width and $\epsilon$-net cardinality:
\begin{equation}\label{Gaussian_width} 
c\cdot {\epsilon}\cdot \sqrt{\log(\#\mathcal{N}(T,\epsilon))} \leq \omega(T) \leq C\cdot \int_{0}^\infty \sqrt{\log \#\mathcal{N}(T,\epsilon)}d\epsilon,
\end{equation} 
where $C$ and $c$ are positive absolute constants. This relationship proves particularly useful in various contexts. For instance, when analyzing $k$-sparse signals within the $n$-dimensional unit ball, defined as 
\begin{equation}
\label{eqn: T}
T := \{\boldsymbol{x}\in \mathbb{R}^{n} : \|\boldsymbol{x}\|_0\leq k,\ \|\boldsymbol{x}\|_2\leq 1\},
\end{equation}
 we find that 
\[
\omega(T) = O(\sqrt{k\log(n/k)}).
\] 
For more examples and details on the estimation of $\omega(T)$, we refer the readers to \cite{one-bit}.

The bi-Lipschitz constants described in (\ref{original_lipschitz2}) are examined in Theorem \ref{thm: lipschitz_result}.

		\begin{theorem}\label{thm: lipschitz_result}
Let $S\subset \mathbb{R}^n$ be a cone. Assume that  $\boldsymbol{A}\in \mathbb{R}^{m\times n}$ is a standard Gaussian random matrix.  Assume that $\delta$ is a fixed positive constant such that $\delta<1/2$. Then with probability at least $1-2\exp(-{c}_{\delta} \cdot m)$, we have: 
\begin{equation}
\label{eqn: final_lip_LU}
\Big(\frac{1}{2}-\phi(\vx,\vy)-\frac{\delta}{2}\Big)\cdot\|\boldsymbol{x}-\boldsymbol{y}\|_2^2\leq \frac{1}{{m}}\cdot \|\sigma(\boldsymbol{A}\boldsymbol{x})-\sigma(\boldsymbol{A}\boldsymbol{y})\|_2^2\leq \Big(\frac{1}{2}-\phi(\vx,\vy)+\frac{\delta}{2}\Big)\cdot\|\boldsymbol{x}-\boldsymbol{y}\|_2^2,
\end{equation}
for any $\boldsymbol{x},\boldsymbol{y}\in S$, provided that 
\[
m\geq {C}_{\delta}\cdot \omega^2((S-S)\cap \mathbb{B}^n).
\]
Here ${c}_{\delta}$ and ${C}_{\delta}$ are positive constants depending on $\delta$, and $\phi(\vx,\vy)$ is defined in (\ref{eqn: phi}).
		\end{theorem}

		\begin{proof}
		The proof is presented in Section \ref{sec: real_lipschitz}. 
		\end{proof}
		\begin{remark}
		\label{rem: key_rem}
		In this paper, we do not optimize the values of $c_{\delta}$ and $C_{\delta}$. We leave it for interested readers. Compared to the result with those in \cite[Theorem 3]{ref1}, our work improves the sampling complexity. Specifically, we have reduced the sampling dependency from $\omega^4(\cdot)$ to $\omega^2(\cdot)$. For instance, consider the $k$-sparse signal set $T$ in (\ref{eqn: T}).
		 Let $S$ be the cone generated by $T$, i.e.,
		\[
		S:=\mathrm{cone}(T):=\{\lambda \vx\ :\ \vx\in T\ \mathrm{and}\ \lambda>0\}. 
		\]
		Our result yields a sampling complexity on the order of $k \log(n/k)$, which is significantly lower than the $k^2\log(n/k)$
  requirement in \cite[Theorem 3]{ref1}.
		
		Furthermore, we contrast our findings with those in \cite[Theorem 4]{ref1}. By setting $\vy = \boldsymbol{0}$, we obtain, with high probability, the following inequality for any $\vx \in S$:

\[
\left( \frac{1}{2} - \frac{\delta}{2} \right) \cdot \|\boldsymbol{x}\|_2^2 \leq \frac{1}{m}\cdot  \|\sigma(\boldsymbol{A} \boldsymbol{x})\|_2^2 \leq \left( \frac{1}{2} + \frac{\delta}{2} \right) \cdot \|\boldsymbol{x}\|_2^2.
\]
Consequently, by directly applying the results from Theorem \ref{thm: lipschitz_result} (details omitted for brevity), we establish that, with high probability, for any nonzero $\vx, \vy \in S$:
\[
\left| \cos \left( \theta_{\sigma(\boldsymbol{A}\boldsymbol{x}), \sigma(\boldsymbol{A}\boldsymbol{y})} \right) - \left( \cos \theta_{\boldsymbol{x}, \boldsymbol{y}} + \frac{\sin \theta_{\boldsymbol{x}, \boldsymbol{y}} - \theta_{\boldsymbol{x}, \boldsymbol{y}} \cdot \cos \theta_{\boldsymbol{x}, \boldsymbol{y}}}{\pi} \right) \right| \lesssim \delta.
\]
 Here $\theta_{\vx,\vy}$ and $\theta_{\sigma(\boldsymbol{A}\boldsymbol{x}), \sigma(\boldsymbol{A}\boldsymbol{y})}$ are defined as in (\ref{eqn: theta}).
In contrast to Theorem 4 of \cite{ref1}, which assumes that $K \subset \mathbb{B}^{n} \setminus \mathbb{B}^{n}_{\beta}$, where $\mathbb{B}^{n}_{\beta} := \left\{ \vx \in \mathbb{R}^n : \|\vx\|_2 \leq \beta \right\}$, our result can be generalized for any $K \subset \mathbb{B}^{n}$ by setting $S$ as 
\[
S: = \mathrm{cone}(K) := \left\{\lambda  \vx\  :\  \vx \in K\  \mathrm{and}\   \lambda > 0 \right\}.
\]
		\end{remark}
	Now we begin to prove Theorem \ref{thm: Gaussian_bilip}. 
		\begin{proof}[Proof of Theorem \ref{thm: Gaussian_bilip}]
      We assert that:
		\begin{equation}\label{eqn: phi_up_low}
		0\leq \phi(\vx,\vy)\leq \frac{1}{4},
		\end{equation}
		for any non-zero $\vx,\vy$ such that $\vx\neq \vy$. The proof of this assertion will be provided at the end. Applying Theorem \ref{thm: lipschitz_result} with $S=\mathbb{R}^{n}$, we can directly infer that for any given $\delta\leq 1/5$, with probability at least $1-\exp(-c_{\delta}\cdot n)$:
		\[
		\mathcal{U}_{\vA,\boldsymbol{0}}\leq \sqrt{\frac{1}{2}+\frac{\delta}{2}},\quad \text{and}\quad \mathcal{L}_{\vA,\boldsymbol{0}}\geq \sqrt{\frac{1}{4}-\frac{\delta}{2}},
		\]
		provided that $m\geq C_{\delta}\cdot n$. 
		
		Therefore, 
		\begin{equation}
		\label{eqn: 4delta}
		\beta_{\vA,\boldsymbol{0}}\leq \sqrt{\frac{\frac{1}{2}+\frac{\delta}{2}}{\frac{1}{4}-\frac{\delta}{2}}}\leq \sqrt{2}+4\delta,
		\end{equation}
		as 
		\[
		\begin{aligned}
		(\sqrt{2}+4\delta)^2\cdot (1/4-\delta/2)-(1/2+\delta/2)\geq &(2+10\delta)\cdot (1/4-\delta/2)-(1/2+\delta/2)=\delta-5\delta^2\geq 0,
		\end{aligned}
		\]
		given that $\delta\leq 1/5$. The conclusion in (\ref{eqn: beta_upper}) holds if we replace $4\delta$ in (\ref{eqn: 4delta}) with $\delta$.

Now, we present the proof of the result in (\ref{eqn: phi_up_low}). Through direct calculation, for any non-zero $\vx,\vy\in \mathbb{R}^{n}$
  such that $\vx\neq \vy$, we have the following property: 
		\begin{equation}\label{eqn: phi_property}
		\begin{aligned}
              0\leq & \phi(\vx,\vy)=\frac{\left(\sin \theta_{\vx,\vy}-\theta_{\vx,\vy}\cdot \cos \theta_{\vx,\vy}\right)}{\pi}\cdot \frac{\|\vx\|_2\cdot \|\vy\|_2}{\|\vx\|_2^2+\|\vy\|_2^2-2\langle \vx,\vy\rangle}\\
              = &\frac{\left(\sin \theta_{\vx,\vy}-\theta_{\vx,\vy}\cdot \cos \theta_{\vx,\vy}\right)}{\pi}\cdot \frac{1}{\frac{\|\vx\|_2}{\|\vy\|_2}+\frac{\|\vy\|_2}{\|\vx\|_2}-2\cos\theta_{\vx,\vy}}\\
              \leq &\frac{\left(\sin \theta_{\vx,\vy}-\theta_{\vx,\vy}\cdot \cos \theta_{\vx,\vy}\right)}{\pi}\cdot \frac{1}{2-2\cos\theta_{\vx,\vy}}\\
              \leq & \frac{1}{4}.
              \end{aligned}
		\end{equation}
The second inequality is derived from the arithmetic-geometric mean inequality, which states that $\frac{\|\boldsymbol{x}\|_2}{\|\boldsymbol{y}\|_2}+\frac{\|\boldsymbol{y}\|_2}{\|\boldsymbol{x}\|_2}\geq 2$ for any nonzero vectors $\boldsymbol{x},\boldsymbol{y}\in \mathbb{R}^{n}$. The final inequality is based on the monotonicity of the function $g(t):=\frac{\sin t-t\cdot \cos t}{1-\cos t}$, which is strictly increasing on the interval $(0,\pi]$, and its limiting behavior as $t$ approaches zero is $\lim_{t\rightarrow 0^+} \frac{\sin t-t\cdot \cos t}{1-\cos t}=0$. Thus, the proof of (\ref{eqn: phi_up_low}) is complete.
		\end{proof}

		\section{Proof of Theorem \ref{thm: lipschitz_result}}\label{sec: real_lipschitz}

Theorem \ref{thm: lipschitz_result} is a direct consequence of two distinct results below, Theorem \ref{thm: large} and Theorem \ref{thm: small}, each employing different technical tools and addressing a specific regime: one for large distance and another for small distance between vectors $\boldsymbol{x},\boldsymbol{y}\in S$, respectively. 

 Initially, we present the behavior of the bi-Lipschitz constants in the regime where the Euclidean distance between the vectors $\boldsymbol{x},\boldsymbol{y}\in S$ is  large. 

\begin{theorem}\label{thm: large}
		Let $S\subseteq \mathbb{R}^n$ be a cone. Assume that  $\boldsymbol{A}\in \mathbb{R}^{m\times n}$ is a standard Gaussian random matrix. For any fixed positive constants $C$ and $\delta$ such that $C<1$ and $\delta<C^2/2$, with probability at least $1-2\exp(-{c}_{\delta}\cdot m)$, we have:
\begin{equation}
\label{eqn: conc_C}
\Big(\frac{1}{2}-\phi(\vx,\vy)-\delta\Big)\cdot\|\boldsymbol{x}-\boldsymbol{y}\|_2^2\leq \frac{1}{m}\cdot \|\sigma(\boldsymbol{A}\boldsymbol{x})-\sigma(\boldsymbol{A}\boldsymbol{y})\|_2^2\leq \Big(\frac{1}{2}-\phi(\vx,\vy)+\delta\Big)\cdot\|\boldsymbol{x}-\boldsymbol{y}\|_2^2,
\end{equation}
for any $\boldsymbol{x},\boldsymbol{y}\in S$ such that  $\|\boldsymbol{x}-\boldsymbol{y}\|_2\geq C\cdot \max\{\|\vx\|_2,\|\vy\|_2\}$, provided that 
\[
m\geq C_{\delta}\cdot \omega^2((S-S)\cap \mathbb{B}^{n}).
\]
Here ${c}_{\delta}$ and $C_{\delta}$ are positive constants depending on $\delta$, and $\phi(\vx,\vy)$ is defined in (\ref{eqn: phi}). 
		\end{theorem}
		\begin{proof}
		The proof is presented in Section \ref{prof_large_sec}.
		\end{proof}

The following theorem addresses the scenario where the distance between $\boldsymbol{x},\boldsymbol{y}\in S$ is small. In this case, traditional analytical tools are insufficient, and the result may be of independent interest to readers.
			\begin{theorem}\label{thm: small}
		Let $S\subset \mathbb{R}^n$ be a cone. Assume that  $\boldsymbol{A}\in \mathbb{R}^{m\times n}$ is a standard Gaussian random matrix.  For any fixed positive constants $\alpha,\delta,C \in (0,1)$ and $\beta\geq 10$ such that $\alpha-C\beta>0$,  with probability at least $1-2\exp(-{c}_{\alpha,\beta,\delta} \cdot m)$, we have: 
\begin{equation}\label{eqn: small_C}
\small
\Big(\frac{1}{2}-3\alpha- \frac{3}{\beta}-2\delta\Big)\cdot\|\boldsymbol{x}-\boldsymbol{y}\|_2^2\leq \frac{1}{m}\cdot \|\sigma(\boldsymbol{A}\boldsymbol{x})-\sigma(\boldsymbol{A}\boldsymbol{y})\|_2^2\leq \Big(\frac{1}{2}+3\alpha+\frac{3}{\beta}+2\delta\Big)\cdot\|\boldsymbol{x}-\boldsymbol{y}\|_2^2,
\end{equation}
for any $\boldsymbol{x},\vy\in S$ such that $\|\boldsymbol{x}-\boldsymbol{y}\|_2\leq C\cdot \max\{\|\vx\|_2,\|\vy\|_2\}$, provided that 
\[m\geq {C}_{\alpha,\beta,\delta}\cdot\omega^2((S-S)\cap \mathbb{B}^{n}).\] 
Here ${c}_{\alpha,\beta,\delta}$ and ${C}_{\alpha,\beta,\delta}$ are positive constants depending on $\alpha,\beta$ and $\delta$. 
\end{theorem}
\begin{proof}
The proof is presented in Section \ref{proof of small}.
\end{proof}
		\begin{remark}
		 Theorem \ref{thm: small} explicitly states that:
		\begin{equation}\label{eqn: good_case}
\lim_{\boldsymbol{y}\rightarrow \boldsymbol{x}}\left|\frac{\|\sigma(\boldsymbol{A}\boldsymbol{x})-\sigma(\boldsymbol{A}\boldsymbol{y})\|_2^2}{\|\boldsymbol{x}-\boldsymbol{y}\|_2^2}-\frac{1}{2}\right|\leq 3\alpha+\frac{3}{\beta}+2\delta.
\end{equation}
In contrast to equation (\ref{wrong_lipschitz}), which diverges as $\|\vx - \vy\|_2$ approaches zero, our result demonstrates that the ratio $\frac{\|\sigma(\boldsymbol{A}\boldsymbol{x})-\sigma(\boldsymbol{A}\boldsymbol{y})\|_2^2}{\|\boldsymbol{x}-\boldsymbol{y}\|_2^2}$ converges to approximately $1/2$. This finding provides a more solid theoretical foundation for determining the condition number of standard Gaussian random matrices.
		\end{remark}

Building upon the results established in Theorems \ref{thm: large} and \ref{thm: small}, we derive the bi-Lipschitz constants for any arbitrary given cone $S$ in Theorem \ref{thm: lipschitz_result}. This provides a generalized framework for determining the condition number $\beta_{\vA,\vb}$,  when $\vb=\boldsymbol{0}$.
\begin{proof}[The proof of Theorem \ref{thm: lipschitz_result}]
In the trivial case where $\boldsymbol{x}=\boldsymbol{y}$, the conclusion holds immediately. Let us now concentrate on the non-trivial scenario in which  $\boldsymbol{x}\neq \boldsymbol{y}$. By symmetry of $\vx$ and $\vy$, we can assume that $\|\boldsymbol{y}\|_2\leq \|\boldsymbol{x}\|_2$ and $\boldsymbol{x}\neq \boldsymbol{0}$. Normalizing both vectors by $\|\vx\|_2$, that is, replacing $\boldsymbol{y}$ with $\frac{\boldsymbol{y}}{\|\boldsymbol{x}\|_2}$ and $\boldsymbol{x}$ with $\frac{\boldsymbol{x}}{\|\boldsymbol{x}\|_2}$, the conclusions also hold. Therefore,  it suffices to examine the case where $\boldsymbol{x}\in S\cap \mathbb{S}^{n-1}$
  and $\boldsymbol{y}\in S\cap\mathbb{B}^{n}$.

For any fixed $\delta \in (0,(1/120)^4)$, we define constants $C = 3\sqrt{\delta}$, $\alpha = \delta^{1/4}$, and $\beta = \frac{1}{6}\delta^{-1/4}$. 
These parameter choices ensure that $C < 1$, $\delta < C^2/2$, $\alpha, \delta \in (0,1)$, $\alpha - C\beta > 0$, and $\beta \geq 10$, thereby satisfying the conditions in Theorem \ref{thm: large} and Theorem \ref{thm: small}. 

On one hand, for any $\vx\in S\cap \mathbb{S}^{n-1}$ and $\vy\in S\cap\mathbb{B}^{n}$ such that $\|\vx-\vy\|_2\geq C$, applying Theorem \ref{thm: large}, we have 
\[
\Big(\frac{1}{2}-\phi(\vx,\vy)-30\delta^{1/4}\Big)\cdot\|\boldsymbol{x}-\boldsymbol{y}\|_2^2\leq \frac{1}{m}\cdot \|\sigma(\boldsymbol{A}\boldsymbol{x})-\sigma(\boldsymbol{A}\boldsymbol{y})\|_2^2\leq \Big(\frac{1}{2}-\phi(\vx,\vy)+30\delta^{1/4}\Big)\cdot\|\boldsymbol{x}-\boldsymbol{y}\|_2^2.
\]

On the other hand, for any $\vx\in S\cap \mathbb{S}^{n-1}$ and $\vy\in S\cap \mathbb{B}^{n}$ such that $\|\vx-\vy\|_2\leq C$, we claim that 
\begin{equation}
\label{eqn: phi_est_final}
0\leq \phi(\vx,\vy)\leq 6\delta^{1/4}.
\end{equation}
Then following from Theorem \ref{thm: small} and 
$
3\alpha +\frac{3}{\beta}+2\delta\leq 23\delta^{1/4},
$
we also obtain that 
\[
\begin{aligned}
\Big(\frac{1}{2}-\phi(\vx,\vy)-30\delta^{1/4}\Big)\cdot\|\boldsymbol{x}-\boldsymbol{y}\|_2^2\leq &\Big(\frac{1}{2}-30\delta^{1/4}\Big)\cdot\|\boldsymbol{x}-\boldsymbol{y}\|_2^2\\
\leq &\frac{1}{m}\cdot\|\sigma(\boldsymbol{A}\boldsymbol{x})-\sigma(\boldsymbol{A}\boldsymbol{y})\|_2^2\leq \Big(\frac{1}{2}+23\delta^{1/4}\Big)\cdot\|\boldsymbol{x}-\boldsymbol{y}\|_2^2\\
\leq &\Big(\frac{1}{2}-\phi(\vx,\vy)+30\delta^{1/4}\Big)\cdot\|\boldsymbol{x}-\boldsymbol{y}\|_2^2.
\end{aligned}
\]
Consequently, we have:
\begin{equation}
\label{eqn: rip_origin}
\Big(\frac{1}{2}-\phi(\vx,\vy)-30\delta^{1/4}\Big)\cdot\|\boldsymbol{x}-\boldsymbol{y}\|_2^2\leq \frac{1}{m}\cdot\|\sigma(\boldsymbol{A}\boldsymbol{x})-\sigma(\boldsymbol{A}\boldsymbol{y})\|_2^2\leq \Big(\frac{1}{2}-\phi(\vx,\vy)+30\delta^{1/4}\Big)\cdot\|\boldsymbol{x}-\boldsymbol{y}\|_2^2
\end{equation}
for any $\vx\in S\cap \mathbb{S}^{n-1}$ and $\vy\in S\cap \mathbb{B}^{n}$, provided that 
\[
m\geq {C}_{\delta}\cdot\omega^2((S-S)\cap \mathbb{B}^{n}).
\]
Here ${c}_{\delta}$ and ${C}_{\delta}$ are positive constants depending on $\delta$. 

By substituting $\delta$ for $60\delta^{1/4}$ in (\ref{eqn: rip_origin}), we can conclude that for any fixed $\delta < \frac{1}{2}$:
\[
\Big(\frac{1}{2}-\phi(\vx,\vy)-\frac{\delta}{2}\Big)\cdot\|\boldsymbol{x}-\boldsymbol{y}\|_2^2\leq \frac{1}{m}\cdot \|\sigma(\boldsymbol{A}\boldsymbol{x})-\sigma(\boldsymbol{A}\boldsymbol{y})\|_2^2\leq \Big(\frac{1}{2}-\phi(\vx,\vy)+\frac{\delta}{2}\Big)\cdot\|\boldsymbol{x}-\boldsymbol{y}\|_2^2.
\]
This completes the conclusion in (\ref{eqn: final_lip_LU}).

Now the only thing left is to prove (\ref{eqn: phi_est_final}). For any $\vx\in S\cap \mathbb{S}^{n-1}$ and $\vy\in S\cap \mathbb{B}^{n}$ such that $0<\|\vx-\vy\|_2\leq C$, we can directly obtain that $\vy\neq \boldsymbol{0}$. Therefore, $\frac{\|\vx\|_2}{\|\vy\|_2}+\frac{\|\vy\|_2}{\|\vx\|_2}\geq 2$, and 
\[
\begin{aligned}
0\leq \phi(\vx,\vy)=&\frac{\sin \theta_{\vx,\vy} -\theta_{\vx,\vy}\cdot \cos\theta_{\vx,\vy}}{\pi}\cdot \frac{1}{\frac{\|\vx\|_2}{\|\vy\|_2}+\frac{\|\vy\|_2}{\|\vx\|_2}-2\cos\theta_{\vx,\vy}}\\
\leq &\frac{\sin \theta_{\vx,\vy} -\theta_{\vx,\vy}\cdot \cos\theta_{\vx,\vy}}{2\pi(1-\cos\theta_{\vx,\vy})}\leq \frac{\theta_{\vx,\vy} -\theta_{\vx,\vy}\cdot \cos\theta_{\vx,\vy}}{2\pi(1-\cos\theta_{\vx,\vy})}\\
= & \frac{\theta_{\vx,\vy}}{2\pi}\leq  \frac{\sin \theta_{\vx,\vy}}{2\pi}\leq  \frac{C}{2\pi\cdot \sqrt{1-C^2}}\leq 6\sqrt{\delta}\leq 6\delta^{1/4},
\end{aligned}
\]
given that $C=3\sqrt{\delta}$ and $\delta\in (0,(1/120)^4)$. Then the proof of (\ref{eqn: phi_est_final}) is complete. 
		\end{proof}

\section{Proof of Theorem \ref{thm: large}}\label{prof_large_sec}
We commence our analysis by introducing several auxiliary results that form the foundation for our subsequent investigation.
\begin{theorem}\label{thm: Bernstein}\cite[Theorem 2.8.1]{vershynin}
 Let $\{x_i\}_{i=1}^n$ be a sequence of independent random variables satisfying $\max_{i}\|x_i\|_{\psi_1}\leq K$, where 
 \begin{equation}
 \label{eqn: psi_1}
 \|\cdot\|_{\psi_1}:=\sup_{p\geq 1}p^{-1}(\mathbb{E}|\cdot|^p)^{1/p}
 \end{equation}
  denotes the sub-exponential norm. Then, for every $t>0$, \begin{equation} \mathbb{P}\left(\Big|\frac{1}{n}\cdot \sum_{i=1}^{n} {x_i} - \frac{1}{n}\cdot \sum_{i=1}^{n} \mathbb{E}{x_i} \Big| > t\right)\leq 2\exp\left(-c_0 n \min\Big(\frac{t^2}{K^2}, \frac{t}{K}\Big)\right), \end{equation} where $c_0>0$ is an absolute constant. 
 \end{theorem}
	Leveraging Theorem \ref{thm: Bernstein} in conjunction with the concept of $\epsilon$-net in Definition \ref{def:epsilon-net}, we derive the following lemma. As this result follows directly from the aforementioned theorem, we omit the detailed proof.
		\begin{lemma}\label{lem: cover_union}
		Let $\{\va_i\}_{i=1}^m$ be a sequence of independent standard Gaussian random vectors in $\mathbb{R}^n$, i.e., $\va_i \sim \mathcal{N}(\boldsymbol{0}, \boldsymbol{I}_n)$,  for $i = 1, \ldots, m$. Consider two sets $T_1, T_2 \subseteq \mathbb{R}^n$ and a function $f: T_1 \times T_2 \times \mathbb{R}^n \to \mathbb{R}$. Assume that there exists a positive absolute constant $C$ such that
$\max_{\vx \in T_1, \vy \in T_2} \|f(\vx, \vy, \va)\|_{\psi_1} \leq C,$
where $\va \sim \mathcal{N}(\boldsymbol{0}, \boldsymbol{I}_n)$ and $\|\cdot\|_{\psi_1}$ denotes the sub-exponential norm in (\ref{eqn: psi_1}). For any $\delta_0 \in (0,1)$ and $\epsilon > 0$, if
\begin{equation}
\label{eqn: sample_number}
m \geq c_1 \cdot \delta_0^{-2}\cdot\big(\log(\#\mathcal{N}(T_1, \epsilon)) + \log(\#\mathcal{N}(T_2, \epsilon))\big),
\end{equation}
then with probability at least $1 - 2\exp(-{c}'_1\delta_0^2m)$, the following inequality holds:
\[
\Big|\frac{1}{m}\cdot \sum_{i=1}^{m}f(\vx,\vy,\va_i) - \mathbb{E}f(\vx,\vy,\va)\Big| \leq \delta_0,
\]
simultaneously for all $\vx \in \mathcal{N}(T_1, \epsilon)$ and $\vy \in \mathcal{N}(T_2, \epsilon)$. Here, $c_1$ and ${c}'_1$ are  positive absolute constants, and $\mathcal{N}(T, \epsilon)$ denotes the $\epsilon$-net of the set $T$. When $T_2=\emptyset$, the sampling complexity in (\ref{eqn: sample_number}) is reduced to $m \geq c_1\cdot \delta_0^{-2}\cdot \log(\#\mathcal{N}(T_1, \epsilon))$.
		\end{lemma}
		
		Next, we recall the concept of Gaussian width in Definition \ref{eqn: gaussian_width} and present Sudakov's inequality:
		\begin{lemma}\cite[Corollary 7.4.3]{vershynin}
		\label{lem: sudakov}
		Let $T\subseteq \mathbb{R}^{n}$. Then, for any $\epsilon>0$, we have 
		\[
		\omega(T)\geq c_2\cdot \epsilon\cdot \sqrt{\log(\#\mathcal{N}(T,\epsilon))},
		\]
		where $c_2$ is an absolute constant, and $\omega(T)$ denotes the Gaussian width of $T$.
		\end{lemma}
Lastly, we present a lemma demonstrating the Restricted Isometry Property (RIP) for Gaussian random matrices, which builds upon the concept of Gaussian width. This result follows directly from a modification of Theorem 9.2 in \cite{CS} and the application of Sudakov's minoration inequality as stated in Lemma \ref{lem: sudakov}. It is worth noting that a similar result can be observed in Lemma 6.1 of \cite{relu_gradient}.
	\begin{lemma}\label{lem: RIP}
	Let $S\subseteq \mathbb{R}^n$ be a cone.  Consider a standard Gaussian random matrix $\boldsymbol{A}\in \mathbb{R}^{m\times n}$. For any fixed $\delta \in (0, 1)$, there exists a positive absolute constant $c_3$ such that, with probability at least $1-2\exp(-c_3 \delta^2 m)$, the following inequality holds:
\begin{equation} (1-\delta)\cdot\|\vx-\vy\|_2^2 \leq \frac{1}{m}\cdot \|\boldsymbol{A}(\vx-\vy)\|_2^2 \leq (1+\delta)\cdot\|\vx-\vy\|_2^2, \end{equation}
for all $\vx,\vy\in S$, provided that
\begin{equation} m \gtrsim \delta^{-4}\cdot \omega^2((S-S)\cap \mathbb{B}^{n}). \end{equation}
Here, $\omega(\cdot)$ denotes the Gaussian width, $\mathbb{B}^{n}$ is the unit ball in $\mathbb{R}^n$, and $S-S$ is defined in (\ref{eqn: S-S}). 
		\end{lemma}
		
Having established these fundamental results, we are now prepared to proceed with the proof of Theorem \ref{thm: large}.
		\begin{proof}[Proof of Theorem \ref{thm: large}]
		Initially, if $\boldsymbol{x}=\boldsymbol{y}=\boldsymbol{0}$, the conclusion holds trivially. We now focus on the case where either $\boldsymbol{x}\neq \boldsymbol{0}$ or $\boldsymbol{y}\neq \boldsymbol{0}$. Without loss of generality, due to the symmetry between $\boldsymbol{x}$ and $\boldsymbol{y}$, we can assume that $\|\boldsymbol{y}\|_2\leq \|\boldsymbol{x}\|_2$ and $\boldsymbol{x}\neq \boldsymbol{0}$. By normalizing both vectors by $\|\vx\|_2$, that is, replacing $\boldsymbol{y}$ with $\frac{\boldsymbol{y}}{\|\boldsymbol{x}\|_2}$ and $\boldsymbol{x}$ with $\frac{\boldsymbol{x}}{\|\boldsymbol{x}\|_2}$, the condition $\|\boldsymbol{x}-\boldsymbol{y}\|_2\geq C\cdot \max\{\|\boldsymbol{x}\|_2,\|\boldsymbol{y}\|_2\}$ can be  transformed into $\|\frac{\boldsymbol{x}}{\|\boldsymbol{x}\|_2}-\frac{\boldsymbol{y}}{\|\boldsymbol{x}\|_2}\|_2\geq C$, while the conclusion remains invariant.

Consequently, in the subsequent analysis, we need only consider the scenario where $\boldsymbol{x}\in S\cap \mathbb{S}^{n-1}$ and $\boldsymbol{y}\in S\cap\mathbb{B}^{n}$, subject to the condition $\|\boldsymbol{x}-\boldsymbol{y}\|_2\geq C$.

				 Denote $\psi:\mathbb{R}^{n}\times \mathbb{R}^{n}\rightarrow \mathbb{R}$ as 
				\[
				\psi(\vx,\vy)=\phi(\vx,\vy)\cdot \|\vx-\vy\|_2^2,
				\]
				 where $\phi(\vx,\vy)$ is defined in (\ref{eqn: phi}). 
				 
				 Let $C<1$. We claim that, for any fixed positive constant $\delta_0$ such that $\sqrt{\delta_0}<C^2/120$, with probability at least $1-4\exp\left(-\widetilde{c}_0\delta_0^2m\right)$, the following inequalities hold:
\begin{equation} 
\label{eqn: final_temp1}
\frac{1}{2}\|\boldsymbol{x}-\boldsymbol{y}\|_2^2-\psi(\vx,\vy)- 60\sqrt{\delta_0} \leq \frac{1}{m}\cdot \|\sigma(\boldsymbol{A}\boldsymbol{x})-\sigma(\boldsymbol{A}\boldsymbol{y})\|_2^2 \leq \frac{1}{2} \|\boldsymbol{x}-\boldsymbol{y}\|_2^2-\psi(\vx,\vy)+60\sqrt{ \delta_0},
  \end{equation}
for all $\boldsymbol{x} \in S \cap \mathbb{S}^{n-1}$ and $\boldsymbol{y} \in S \cap \mathbb{B}^{n}$ satisfying $\|\boldsymbol{x}-\boldsymbol{y}\|_2 \geq C$, provided that
\begin{equation} 
\label{eqn: sampling} 
m \geq \widetilde{c}_1 \cdot {\delta_0^{-4}}\cdot {\omega^2((S-S) \cap \mathbb{B}^{n})}.
 \end{equation}
 Here $\widetilde{c}_0$ and $\widetilde{c}_1$ are positive absolute constants. 
 
Let $\delta$ and $\delta_0$ be chosen such that  $\delta \leq C^2/2$ and  $\sqrt{\delta_0}=\delta^2/60$. It follows immediately that $\sqrt{\delta_0}<C^2/120$, which satisfies the condition for (\ref{eqn: final_temp1}).  Consequently, (\ref{eqn: final_temp1}) can be rewritten as:
\begin{equation} \label{eqn: final_temp2}
\Big( \frac{1}{2}-\phi(\vx,\vy)\Big)\cdot \|\boldsymbol{x}-\boldsymbol{y}\|_2^2- \delta^2 \leq \frac{1}{m}\cdot\|\sigma(\boldsymbol{A}\boldsymbol{x})-\sigma(\boldsymbol{A}\boldsymbol{y})\|_2^2 \leq \Big( \frac{1}{2}-\phi(\vx,\vy)\Big)\cdot \|\boldsymbol{x}-\boldsymbol{y}\|_2^2 + \delta^2, \end{equation}
with probability at least $1-4\exp\left(-{c}_{\delta}\cdot m\right)$. Besides,  the sampling complexity becomes $m \geq C_{\delta}\cdot \omega^2((S-S) \cap \mathbb{B}^{n})$. Here  ${c}_{\delta}=\widetilde{c}_0\cdot \delta^4/(60^4)$ and ${C}_{\delta}=\widetilde{c}_1\cdot (60)^{16}/\delta^{16}$.

Moreover, given that $\|\boldsymbol{x}-\boldsymbol{y}\|_2 \geq C$ and $\delta \leq C^2/2$, we have:
\[
\delta^2 \leq \delta^2 \cdot \frac{\|\boldsymbol{x}-\boldsymbol{y}\|_2^2}{C^2} \leq \delta \cdot \|\boldsymbol{x}-\boldsymbol{y}\|_2^2.
\]
Thus, \eqref{eqn: final_temp2} can be further simplified to:
\[
\left(\frac{1}{2}-\phi(\vx,\vy)-\delta\right) \cdot \|\boldsymbol{x}-\boldsymbol{y}\|_2^2 \leq \frac{1}{m}\cdot\|\sigma(\boldsymbol{A}\boldsymbol{x})-\sigma(\boldsymbol{A}\boldsymbol{y})\|_2^2 \leq \left(\frac{1}{2}-\phi(\vx,\vy)+\delta\right) \cdot \|\boldsymbol{x}-\boldsymbol{y}\|_2^2.
\]
This establishes the validity of the conclusion in (\ref{eqn: conc_C}).
 
To prove \eqref{eqn: final_temp1} under the sampling complexity in \eqref{eqn: sampling}, we denote the  sets:
\begin{equation} \label{eqn: T} T_1:=S\cap \mathbb{S}^{n-1}\qquad \text{and}\qquad T_2:=S\cap \mathbb{B}^{n}. \end{equation}
Let $\mathcal{N}(T_1,\delta_0)$ and $\mathcal{N}(T_2,\delta_0)$ denote the $\delta_0$-nets of $T_1$ and $T_2$, respectively. 
The subsequent proof is divided into three steps:
 
 \textbf{Step 1: Demonstrate the concentration behavior for all $(\vx_{\delta_0},\vy_{\delta_0})\in \mathcal{N}(T_1,\delta_0)\times \mathcal{N}(T_2,\delta_0)$.}
 
Let us define $f: T_1 \times T_2 \times \mathbb{R}^n \rightarrow \mathbb{R}$ as:
\begin{equation} \label{eqn: f_1} f(\vx,\vy,\va):=(\sigma(\langle \va,\vx\rangle)-\sigma(\langle \va,\vy\rangle))^2. \end{equation}
Through direct computation, we can establish that for any $(\vx,\vy) \in T_1 \times T_2$, $\|f(\vx,\vy,\va)\|_{\psi_1} \leq C$ for some absolute constant $C$, where  $\va\sim \mathcal{N}(\boldsymbol{0},\boldsymbol{I}_n)$. 

 Take $\va_i\in \mathbb{R}^{n}$ as the $i$-th row of matrix $\vA$, for $i=1,\ldots,m$. Consequently, by applying Lemma \ref{lem: cover_union} to the sets $T_1$ and $T_2$ in \eqref{eqn: T} and the function $f$  in \eqref{eqn: f_1}, we can assert that for any $(\vx_{\delta_0},\vy_{\delta_0}) \in \mathcal{N}(T_1,\delta_0) \times \mathcal{N}(T_2,\delta_0)$, with probability at least
\begin{equation} \label{eqn: c0} 1-2\exp(-c'_2\delta_0^2m), \end{equation}
the following concentration inequality holds:
\begin{equation} \label{eqn: f_concentration}
 -\delta_0+\sum_{i=1}^{m}\mathbb{E}f(\vx_{\delta_0},\vy_{\delta_0},\va_i) \leq \frac{1}{m}\cdot \sum_{i=1}^{m}f(\vx_{\delta_0},\vy_{\delta_0},\va_i) \leq \sum_{i=1}^{m}\mathbb{E}f(\vx_{\delta_0},\vy_{\delta_0},\va_i)+\delta_0, 
 \end{equation}
uniformly for all pairs $(\vx_{\delta_0},\vy_{\delta_0}) \in \mathcal{N}(T_1,\delta_0) \times \mathcal{N}(T_2,\delta_0)$, provided that 
\begin{equation}
\label{eqn: sample_temp}
m \geq c_2\cdot  \delta_0^{-2}\cdot \big(\log(\#\mathcal{N}(T_1,\delta_0))+\log(\#\mathcal{N}(T_2,\delta_0))\big).
\end{equation}
 Here, $c_2$ and $c_2'$ are positive absolute constants.

In conjunction with the expectation bound established in \cite[Equality (7)]{ref2}, that is,
		 \begin{equation}\label{eqn: expectation}
 \mathbb{E} f(\vx_{\delta_0},\vy_{\delta_0},\va)= \frac{1}{2}\|\boldsymbol{x}_{{\delta_0}}-\boldsymbol{y}_{{\delta_0}}\|_2^2-\psi(\vx_{\delta_0},\vy_{\delta_0}),
	\end{equation}
we can reformulate (\ref{eqn: f_concentration}) as:
		\begin{equation}\label{eqn: X_epsilon_concen}
		 \frac{1}{2}\|\boldsymbol{x}_{{\delta_0}}-\boldsymbol{y}_{{\delta_0}}\|_2^2-\psi(\vx_{\delta_0},\vy_{\delta_0})- \delta_0\leq \frac{1}{m}\cdot \|\sigma(\boldsymbol{A}\boldsymbol{x}_{{\delta_0}})-\sigma(\boldsymbol{A}\boldsymbol{y}_{{\delta_0}})\|_2^2\leq   \frac{1}{2}\|\boldsymbol{x}_{{\delta_0}}-\boldsymbol{y}_{{\delta_0}}\|_2^2-\psi(\vx_{\delta_0},\vy_{\delta_0})+ \delta_0,
		\end{equation}
	for all $(\vx_{{\delta_0}},\vy_{{\delta_0}})\in \mathcal{N}(T_1,\delta_0) \times \mathcal{N}(T_2,\delta_0)$. The sampling complexity in (\ref{eqn: sample_temp}) becomes  $m\gtrsim \delta_0^{-4}\cdot \omega^2((S-S)\cap \mathbb{B}^n)$, which is satisfied due to the following chain of inequalities:
\[
\begin{aligned}
\omega((S-S)\cap \mathbb{B}^{n})\geq \omega(S\cap \mathbb{B}^{n})\overset{(a)}\gtrsim& \delta_0\cdot\sqrt{\log(\#\mathcal{N}(S\cap \mathbb{B}^n,{\delta_0}))}\\
\overset{(b)}\gtrsim& \delta_0\cdot \sqrt{\log(\#\mathcal{N}(T_1,{\delta_0}))+\log(\#\mathcal{N}(T_2,{\delta_0}))}.
\end{aligned}
\]
	Here, $(a)$ is based on {the Sudakov's inequality} as stated in Lemma \ref{lem: sudakov}, and $(b)$ is a consequence of the following bound:
	\[
	\mathcal{N}(T_1,\delta_0) \times \mathcal{N}(T_2,\delta_0)\leq \mathcal{N}(S\cap \mathbb{B}^{n},{\delta_0})\times \mathcal{N}(S\cap \mathbb{B}^{n},{\delta_0}).
	\]
	
	\textbf{Step 2: Prove the upper bound in  (\ref{eqn: final_temp1}).} 
	For any $\vx\in T_1$ and $\vy\in T_2$, we can take corrsesponding  $\vx_{\delta_0} \in \mathcal{N}(T_1,\delta_0)$ and $\vy_{\delta_0} \in \mathcal{N}(T_2,\delta_0)$ such that $\|\vx-\vx_{\delta_0}\|_2\leq \delta_0$ and $\|\vy-\vy_{\delta_0}\|_2\leq \delta_0$, which leads to:
\begin{equation} \label{eqn: epsilon_net} \|\vx-\vx_{\delta_0}\|_2^2+ \|\vy-\vy_{\delta_0}\|_2^2\leq 2\delta_0^2. \end{equation}
We claim that:
	\begin{equation}\label{eqn: claim1}
 \|\vx_{\delta_0}-\vy_{\delta_0}\|_2^2\leq \|\vx-\vy\|_2^2+12\delta_0,
	\end{equation}
	and 
		\begin{equation}\label{eqn: claim11}
| \psi(\vx_{\delta_0},\vy_{\delta_0})- \psi(\vx,\vy)|\leq 10\sqrt{\delta_0}. 
	\end{equation}
The proofs of them will be provided later.

Applying Lemma \ref{lem: RIP}, we can assert that for any fixed positive constant $\delta_0<C^2/120$, with probability at least $1-4\exp(-\widetilde{c}_2\delta_0^2m)$, where $\widetilde{c}_2$ is also an absolute constant, the following inequalities hold:
	\begin{equation}
	\label{eqn: upper_temp1}
	\begin{aligned}
	&\frac{1}{\sqrt{m}}\cdot \|\boldsymbol{A}(\boldsymbol{x}-\boldsymbol{x}_{{\delta_0}})\|_2 + \frac{1}{\sqrt{m}}\cdot \|\boldsymbol{A}(\boldsymbol{y}_{{\delta_0}}-\boldsymbol{y})\|_2\\
	\leq &\sqrt{1+\delta_0}\cdot (\|\vx-\vx_{\delta_0}\|_2+\|\vy-\vy_{\delta_0}\|_2)\\
	\leq & \sqrt{2(1+\delta_0)}\cdot\sqrt{\|\vx-\vx_{\delta_0}\|_2^2+\|\vy-\vy_{\delta_0}\|_2^2}\leq 3\delta_0,
	\end{aligned}
	\end{equation}
	provided that $m\gtrsim \delta_0^{-4}\cdot \omega^2((S-S)\cap \mathbb{B}^n)$. 
	The final inequality in the above expression is derived from the bound:
	\[
	\sqrt{2(1+\delta_0)}\cdot\sqrt{\|\vx-\vx_{\delta_0}\|_2^2+\|\vy-\vy_{\delta_0}\|_2^2}\leq 3\delta_0,
	\]
	given that $\delta_0\leq 1$.
	
	Consequently,  the upper bound for $\frac{1}{{m}}\cdot \|\sigma(\boldsymbol{A}\boldsymbol{x})-\sigma(\boldsymbol{A}\boldsymbol{y})\|_2^2$ can be derived:
\begin{equation}
\label{eqn: upper_c_large}
			\begin{split}
				& \frac{1}{\sqrt{m}}\cdot \|\sigma(\boldsymbol{A}\boldsymbol{x})-\sigma(\boldsymbol{A}\boldsymbol{y})\|_2 \\
				\leq& \frac{1}{\sqrt{m}}\cdot \|\sigma(\boldsymbol{A}\boldsymbol{x})-\sigma(\boldsymbol{A}\boldsymbol{x}_{{\delta_0}})\|_2 + \frac{1}{\sqrt{m}}\cdot \|\sigma(\boldsymbol{A}\boldsymbol{y}_{{\delta_0}})-\sigma(\boldsymbol{A}\boldsymbol{y})\|_2+\frac{1}{\sqrt{m}}\cdot \|\sigma(\boldsymbol{A}\boldsymbol{x}_{{\delta_0}})-\sigma(\boldsymbol{A}\boldsymbol{y}_{{\delta_0}})\|_2\\
				\overset{(c)}\leq &\frac{1}{\sqrt{m}}\cdot \|\boldsymbol{A}(\boldsymbol{x}-\boldsymbol{x}_{{\delta_0}})\|_2 + \frac{1}{\sqrt{m}}\cdot \|\boldsymbol{A}(\boldsymbol{y}_{{\delta_0}}-\boldsymbol{y})\|_2+\frac{1}{\sqrt{m}}\cdot \|\sigma(\boldsymbol{A}\boldsymbol{x}_{{\delta_0}})-\sigma(\boldsymbol{A}\boldsymbol{y}_{{\delta_0}})\|_2\\
				\overset{(d)}{\leq}& 3\delta_0+\sqrt{ \frac{1}{2}\|\boldsymbol{x}_{{\delta_0}}-\boldsymbol{y}_{{\delta_0}}\|_2^2-\psi(\vx_{\delta_0},\vy_{\delta_0})+\delta_0},
			\end{split}
\end{equation}
where $(c)$ follows from $\|\sigma(\vz_1)-\sigma(\vz_2)\|_2\leq \|\vz_1-\vz_2\|_2$ for all $\vz_1,\vz_2\in \mathbb{R}^{m}$, and $(d)$ is based on (\ref{eqn: upper_temp1}) and (\ref{eqn: X_epsilon_concen}).

Besides, 	substituting (\ref{eqn: claim1}) and (\ref{eqn: claim11}) into (\ref{eqn: upper_c_large}), we obtain:
 \begin{equation}\label{large_upper}
 \begin{split}
 \frac{1}{{m}}\cdot \|\sigma(\boldsymbol{A}\boldsymbol{x})-\sigma(\boldsymbol{A}\boldsymbol{y})\|_2^2\leq &\left(3\delta_0+\sqrt{\frac{1}{2}\|\boldsymbol{x}_{{\delta_0}}-\boldsymbol{y}_{{\delta_0}}\|_2^2-\psi(\vx_{\delta_0},\vy_{\delta_0})+\delta_0}\right)^2\\
 \leq & \left(3 {\delta_0}+\sqrt{\frac{1}{2}\|\boldsymbol{x}-\boldsymbol{y}\|_2^2-\psi(\vx,\vy)+7{\delta_0}+10\sqrt{\delta_0}}\right)^2\\
 \leq & \frac{1}{2}\|\boldsymbol{x}-\boldsymbol{y}\|_2^2-\psi(\vx,\vy)+7\delta_0+10\sqrt{\delta_0}+30\delta_0+9\delta_0^2\\
 \leq &\frac{1}{2}\|\vx-\vy\|_2^2-\psi(\vx,\vy)+60\sqrt{\delta_0}.
 \end{split}
 \end{equation}
 The third line above follows from: 
 \[
\sqrt{\frac{1}{2}\|\vx-\vy\|_2^2-\psi(\vx,\vy)+7\delta_0+10\sqrt{\delta_0}}\leq 5,
 \]
 as $\psi(\vx,\vy)\geq 0$, $\|\vx-\vy\|_2\leq 2$ and $\delta_0\leq 1$. 
Thus, the upper bound in (\ref{eqn: final_temp1}) is established.

To complete the proof, we now demonstrate (\ref{eqn: claim1}) and (\ref{eqn: claim11}). By direct calculations, we have:
 \[
 \|\vx-\vx_{{\delta_0}}\|_2+\|\vy-\vy_{{\delta_0}}\|_2\leq 2\delta_0.
 \] 
 Therefore, 
\begin{equation}
\label{eqn: temp_xdelta_ydelta}
\begin{aligned}
 \|\boldsymbol{x}_{{\delta_0}}-\boldsymbol{y}_{{\delta_0}}\|_2^2= &\|\vx_{\delta_0}-\vx+\vy-\vy_{\delta_0}+\vx-\vy\|_2^2 \\
 \leq & (\|\vx-\vy\|_2+\|\vx-\vx_{{\delta_0}}\|_2+\|\vy-\vy_{{\delta_0}}\|_2)^2
 \leq (\|\boldsymbol{x}-\boldsymbol{y}\|_2+2{\delta_0})^2\\
 \leq & \|\boldsymbol{x}-\boldsymbol{y}\|_2^2+12{\delta_0}.
 \end{aligned}
 \end{equation}
The last line above also utilizes $\|\vx-\vy\|_2\leq 2$ and $\delta_0\leq 1$. This completes the proof of (\ref{eqn: claim1}).
 
We now proceed to demonstrate the result stated in (\ref{eqn: claim11}). Our analysis will be divided into the following two subcases:

				\textbf{Subcase 1: $\|\vy\|_2\leq \sqrt{\delta_0}$.}
				
				Given that
				\[
				\|\vy_{\delta_0}\|_2\leq \|\vy-\vy_{\delta_0}\|_2+\|\vy\|_2\leq \delta_0+\sqrt{\delta_0}.
				\]
				and 
				\[
				\max_{\theta\in [0,\pi]}|\sin\theta-\theta\cdot \cos\theta|\leq \pi,
				\]
				coupled with the fact that $\|\boldsymbol{x}\|_2 = \|\boldsymbol{x}_{\delta_0}\|_2 = 1$, we can deduce that
				\[
				|\psi(\vx,\vy)|\leq \sqrt{\delta_0},\quad \text{and}\quad |\psi(\vx_{\delta_0},\vy_{\delta_0})|\leq \sqrt{\delta_0}+\delta_0.
				\]
				Consequently, by the triangle inequality, we obtain:
				\[
				|\psi(\vx_{\delta_0},\vy_{\delta_0})-\psi(\vx,\vy)|\leq  \delta_0+2\sqrt{\delta_0}\leq 10\sqrt{\delta_0}. 
				\]
				\textbf{Subcase 2: $\|\vy\|_2> \sqrt{\delta_0}$.}  
				
				Given that $\|\boldsymbol{x}\|_2=1$  and $\|\boldsymbol{y}\|_2\geq \sqrt{\delta_0}$, coupled with the conditions $\|\boldsymbol{x}-\boldsymbol{x}_{\delta_0}\|_2\leq \delta_0$ and $\|\boldsymbol{y}-\boldsymbol{y}_{\delta_0}\|_2\leq \delta_0$, we can deduce that
				\[
				\theta_{\vx,\vx_{\delta_0}}\leq \tan (\theta_{\vx,\vx_{\delta_0}})\leq \frac{\delta_0}{\sqrt{1-\delta_0^2}}\leq 2\delta_0\qquad \text{and}\qquad \theta_{\vy,\vy_{\delta_0}}\leq \tan (\theta_{\vy,\vy_{\delta_0}})\leq \frac{\delta_0}{\sqrt{\delta_0-\delta_0^2}}\leq 2\sqrt{\delta_0},
				\]
				as $\delta_0<1/120$. Consequently, 
				\[
				|\theta_{\vx,\vy}-\theta_{\vx_{\delta_0},\vy_{\delta_0}}|\leq 2\delta_0+2\sqrt{\delta_0},
				\]
				which implies
				\begin{equation}
				\label{eqn: sin_theta_temp}
				\begin{aligned}
				&\big|(\sin \theta_{\vx,\vy}-\theta_{\vx,\vy}\cdot \cos \theta_{\vx,\vy})-(\sin \theta_{\vx_{\delta_0},\vy_{\delta_0}}-\theta_{\vx_{\delta_0},\vy_{\delta_0}}\cdot \cos \theta_{\vx_{\delta_0},\vy_{\delta_0}})\big|\\
				\leq &2|\theta_{\vx,\vy}-\theta_{\vx_{\delta_0,\vy_{\delta_0}}}|\leq 4\delta_0+4\sqrt{\delta_0}, 
				\end{aligned}
				\end{equation}
				as the derivative of $f(\theta):=\sin \theta-\theta\cdot \cos \theta$ is $g(\theta):=\theta\sin \theta\in [0,2]$, for  any $\theta\in [0,\pi]$.
				
				Therefore, we can get that 
				\[
				\begin{aligned}
			&|\psi(\vx_{\delta_0},\vy_{\delta_0})-\psi(\vx,\vy)|\\
			\leq &\frac{\Big|(\sin \theta_{\vx,\vy}-\theta_{\vx,\vy}\cdot \cos \theta_{\vx,\vy})-(\sin \theta_{\vx_{\delta_0},\vy_{\delta_0}}-\theta_{\vx_{\delta_0},\vy_{\delta_0}}\cdot \cos \theta_{\vx_{\delta_0},\vy_{\delta_0}})\Big|}{\pi}\cdot \|\vx\|_2\cdot \|\vy\|_2\\
			 &+\frac{\big|\sin \theta_{\vx_{\delta_0},\vy_{\delta_0}}-\theta_{\vx_{\delta_0},\vy_{\delta_0}}\cdot \cos \theta_{\vx_{\delta_0},\vy_{\delta_0}}\big|}{\pi}\cdot \Big|\|\vx\|_2\cdot \|\vy\|_2-\|\vx_{\delta_0}\|_2\cdot \|\vy_{\delta_0}\|_2\Big|\\
			\leq & 5\delta_0+4\sqrt{\delta_0}\leq 10\sqrt{\delta_0}. 
			\end{aligned}
				\]
				The last line above follows from (\ref{eqn: sin_theta_temp}), $\big|\sin \theta_{\vx_{\delta_0},\vy_{\delta_0}}-\theta_{\vx_{\delta_0},\vy_{\delta_0}}\cdot \cos \theta_{\vx_{\delta_0},\vy_{\delta_0}}\big|\leq \pi$, and $\|\vy-\vy_{\delta_0}\|_2\leq \delta_0$. 
				 Thus, the proof of (\ref{eqn: claim11}) is complete.

  \textbf{Step 3: Prove the lower bound in  (\ref{eqn: final_temp1}).} 
 
 Similarly, we begin by asserting the following claim:
	\begin{equation}\label{eqn: claim2}
 \|\vx_{\delta_0}-\vy_{\delta_0}\|_2^2\geq \|\vx-\vy\|_2^2-12\delta_0,
	\end{equation}
	for any  $(\vx,\vy)\in T_1\times T_2$ with corresponding $(\vx_{\delta_0},\vy_{\delta_0})\in \mathcal{N}(T_1,\delta_0)\times \mathcal{N}(T_2,\delta_0)$ such that $\|\vx-\vx_{\delta_0}\|_2\leq \delta_0$ and $\|\vy-\vy_{\delta_0}\|_2\leq \delta_0$.
	Combined with (\ref{eqn: claim11}), we can directly deduce:
	\begin{equation}
	\label{eqn: lower_tempXY}
	\frac{1}{2}\|\vx_{\delta_0}-\vy_{\delta_0}\|_2^2-\psi(\vx_{\delta_0},\vy_{\delta_0})-\delta_0\geq \frac{1}{2}\|\vx-\vy\|_2^2-10\sqrt{\delta_0}-7\delta_0\geq \frac{1}{2}\|\vx-\vy\|_2^2-17\sqrt{\delta_0}>0,
	\end{equation}
	given that $\|\vx-\vy\|_2^2\geq C^2>120\sqrt{\delta_0}$. 
	
Combining (\ref{eqn: lower_tempXY}) with (\ref{eqn: upper_temp1}) and (\ref{eqn: X_epsilon_concen}), we obtain:
 \begin{equation}
 \label{eqn: lower_long}
			\begin{split}
				& \frac{1}{\sqrt{m}}\cdot \|\sigma(\boldsymbol{A}\boldsymbol{x})-\sigma(\boldsymbol{A}\boldsymbol{y})\|_2 \\
				\geq& -\frac{1}{\sqrt{m}}\cdot \|\sigma(\boldsymbol{A}\boldsymbol{x})-\sigma(\boldsymbol{A}\boldsymbol{x}_{{\delta_0}})\|_2 -\frac{1}{\sqrt{m}}\cdot \|\sigma(\boldsymbol{A}\boldsymbol{y}_{{\delta_0}})-\sigma(\boldsymbol{A}\boldsymbol{y})\|_2+\frac{1}{\sqrt{m}}\cdot \|\sigma(\boldsymbol{A}\boldsymbol{x}_{{\delta_0}})-\sigma(\boldsymbol{A}\boldsymbol{y}_{{\delta_0}})\|_2\\
				\geq &-\frac{1}{\sqrt{m}}\cdot \|\boldsymbol{A}(\boldsymbol{x}-\boldsymbol{x}_{{\delta_0}})\|_2 - \frac{1}{\sqrt{m}}\cdot \|\boldsymbol{A}(\boldsymbol{y}_{{\delta_0}}-\boldsymbol{y})\|_2+\frac{1}{\sqrt{m}}\cdot \|\sigma(\boldsymbol{A}\boldsymbol{x}_{{\delta_0}})-\sigma(\boldsymbol{A}\boldsymbol{y}_{{\delta_0}})\|_2\\
				\geq & -3\delta_0+\sqrt{ \frac{1}{2}\|\boldsymbol{x}_{{\delta_0}}-\boldsymbol{y}_{{\delta_0}}\|_2^2-\psi(\vx_{\delta_0},\vy_{\delta_0})- \delta_0}
				\geq -3\delta_0+\sqrt{ \frac{1}{2}\|\boldsymbol{x}-\boldsymbol{y}\|_2^2-\psi(\vx,\vy)- 17\sqrt{\delta_0}}.
			\end{split}
\end{equation}
Therefore, squaring both sides of (\ref{eqn: lower_long}), we derive: 
 \begin{equation}\label{large_lower}
 \begin{split}
 &\frac{1}{{m}}\|\sigma(\boldsymbol{A}\boldsymbol{x})-\sigma(\boldsymbol{A}\boldsymbol{y})\|_2^2\geq \left(-3\delta_0+\sqrt{ \frac{1}{2}\|\boldsymbol{x}-\boldsymbol{y}\|_2^2-\psi(\vx,\vy)- 17\sqrt{\delta_0}}\right)^2\\
\geq& \frac{1}{2}\|\boldsymbol{x}-\boldsymbol{y}\|_2^2-\psi(\vx,\vy)- 17\sqrt{\delta_0}-6\delta_0\cdot \sqrt{\frac{1}{2}\|\boldsymbol{x}-\boldsymbol{y}\|_2^2-\psi(\vx,\vy)- 17\sqrt{\delta_0}}\\
\geq &\frac{1}{2}\|\boldsymbol{x}-\boldsymbol{y}\|_2^2-\psi(\vx,\vy)- 17\sqrt{\delta_0}-6\sqrt{2}\delta_0\geq \frac{1}{2}\|\boldsymbol{x}-\boldsymbol{y}\|_2^2-\psi(\vx,\vy)- 60\sqrt{\delta_0}.
 \end{split}
 \end{equation}
The last line above follows from $\psi(\vx,\vy)\geq 0$ and  $\|\vx-\vy\|_2\leq 2$. Thus, the lower bound in (\ref{eqn: final_temp1}) is established.

To complete the proof, we demonstrate (\ref{eqn: claim2}). By direct calculations:
 \[
 \|\vx-\vx_{{\delta_0}}\|_2+\|\vy-\vy_{{\delta_0}}\|_2\leq \sqrt{2}\delta_0<2\delta_0<2\sqrt{\delta_0}<C^2<C<\|\vx-\vy\|_2.
 \] 
 Therefore, 
\begin{equation}
\label{eqn: temp_xdelta_ydelta}
\begin{aligned}
 \|\boldsymbol{x}_{{\delta_0}}-\boldsymbol{y}_{{\delta_0}}\|_2^2= &\|\vx_{\delta_0}-\vx+\vy-\vy_{\delta_0}+\vx-\vy\|_2^2 \\
 \geq & (\|\vx-\vy\|_2-\|\vx-\vx_{{\delta_0}}\|_2-\|\vy-\vy_{{\delta_0}}\|_2)^2
 \geq (\|\boldsymbol{x}-\boldsymbol{y}\|_2-2{\delta_0})^2\\
 \geq & \|\boldsymbol{x}-\boldsymbol{y}\|_2^2-4{\delta_0}\cdot\|\vx-\vy\|_2\geq \|\vx-\vy\|_2^2-8\delta_0\geq \|\vx-\vy\|_2^2-12\delta_0.
 \end{aligned}
 \end{equation}
The last line above also follows from $\|\vx-\vy\|_2\leq 2$. This completes the proof of (\ref{eqn: claim2}).
	    \end{proof}

		
\section{Proof of Theorem \ref{thm: small}}\label{proof of small}
	\subsection{Auxiliary Expectation Lemmas}
	
	We begin by presenting several expectation results that are instrumental in establishing the main theorem. In the following discussion, $\mathbf{1}_{\{\cdot\}}$ denotes the indicator function, which takes the value 1 when the condition inside the curly braces is true and 0 otherwise.
	
	\begin{lemma} \label{lem: mid} Let $\boldsymbol{a}\in \mathbb{R}^{n}$ be a random vector  such that $\boldsymbol{a} \sim \mathcal{N}(\boldsymbol{0},\boldsymbol{I}_n)$. For any fixed unit vectors $\boldsymbol{x}, \boldsymbol{y} \in \mathbb{S}^{n-1}$ and a non-negative constant $\alpha \in (0,1)$, the following conclusions hold:
	\begin{equation}
	\label{eqn: E1}
	\mathbb{E}\left(\langle \boldsymbol{a}, \boldsymbol{y} \rangle^2\cdot  \boldsymbol{1}_{\{\langle \boldsymbol{a}, \boldsymbol{x} \rangle \geq 0\}} \right)= \frac{1}{2},
	\end{equation}
	and 
	\begin{equation}
	\label{eqn: E2}
	\mathbb{E}\left(\langle \boldsymbol{a}, \boldsymbol{y} \rangle^2\cdot \boldsymbol{1}_{\{\alpha\geq \langle \boldsymbol{a}, \boldsymbol{x} \rangle > 0\}}\right)=\mathbb{E}\left(\langle \boldsymbol{a}, \boldsymbol{y} \rangle^2 \cdot \boldsymbol{1}_{\{-\alpha\leq \langle \boldsymbol{a}, \boldsymbol{x} \rangle < 0\}}\right) \leq  2\alpha.
	\end{equation}

		\end{lemma}
		\begin{proof}
		Leveraging the symmetry of the Gaussian distribution, we have:
		\[
		\mathbb{E}\left(\langle \boldsymbol{a}, \boldsymbol{y} \rangle^2 \cdot \boldsymbol{1}_{\{\alpha\geq \langle \boldsymbol{a}, \boldsymbol{x} \rangle > 0\}}\right)=\mathbb{E}\left(\langle \boldsymbol{a}, \boldsymbol{y} \rangle^2 \cdot \boldsymbol{1}_{\{-\alpha\leq \langle \boldsymbol{a}, \boldsymbol{x} \rangle < 0\}}\right).
		\]
		Thus, it suffices to consider the upper bound of $\mathbb{E}\left(\langle \boldsymbol{a}, \boldsymbol{y} \rangle^2 \cdot \mathbf{1}_{\{0 < \langle \boldsymbol{a}, \boldsymbol{x} \rangle \leq \alpha\}}\right)$ in (\ref{eqn: E2}).

		\textbf{Case 1: $\boldsymbol{x}=\boldsymbol{y}$.} Let $a:=\langle \va,\vx\rangle\sim \mathcal{N}(0,1)$. Then: 
		\[\mathbb{E}\left(\langle \boldsymbol{a}, \boldsymbol{y} \rangle^2\cdot  \boldsymbol{1}_{\{\langle \boldsymbol{a}, \boldsymbol{x} \rangle \geq 0\}}\right) =\mathbb{E}\left(\langle \boldsymbol{a}, \boldsymbol{x} \rangle^2 \cdot \boldsymbol{1}_{\{\langle \boldsymbol{a}, \boldsymbol{x} \rangle \geq 0\}}\right) =\mathbb{E}\left(a^2\cdot \boldsymbol{1}_{\{a\geq 0\}}\right)= \frac{1}{2},\]
		and 
		\[
\mathbb{E}\left(\langle \boldsymbol{a}, \boldsymbol{y} \rangle^2\cdot  \boldsymbol{1}_{\{\alpha\geq \langle \boldsymbol{a}, \boldsymbol{x} \rangle >0\}}\right)=\mathbb{E}\left(\langle \boldsymbol{a}, \boldsymbol{x} \rangle^2\cdot  \boldsymbol{1}_{\{\alpha\geq \langle \boldsymbol{a}, \boldsymbol{x} \rangle > 0\}}\right)=\mathbb{E}\left(a^2\cdot \boldsymbol{1}_{\{\alpha\geq a>0\}}\right)\leq \alpha^2\leq 2\alpha. 
\]
		\textbf{Case 2: $\boldsymbol{x}\neq \boldsymbol{y}$.} We can represent  $\boldsymbol{y}$  as:
		\[
		\boldsymbol{y}=\cos\theta \cdot \boldsymbol{x}+\sin\theta \cdot \boldsymbol{x}_{\perp},
		\]
		where $\cos\theta = \langle \boldsymbol{x}, \boldsymbol{y} \rangle$ and $\sin\theta = \|\boldsymbol{y}-\langle \boldsymbol{x}, \boldsymbol{y} \rangle\cdot \vx\|_2$ for some $\theta\in[0,\pi]$, and $\boldsymbol{x}_{\perp} = \frac{\boldsymbol{y}-\langle \boldsymbol{x}, \boldsymbol{y} \rangle \boldsymbol{x}}{\|\boldsymbol{y}-\langle \boldsymbol{x}, \boldsymbol{y} \rangle \boldsymbol{x}\|_2}$. Let  $a_1:=\langle \va,\vx\rangle$ and $a_2:=\langle \va,\vx_{\perp}\rangle$. Then:
		\[
		\langle \va,\vy \rangle = \cos\theta\cdot \langle \va,\vx\rangle+\sin\theta\cdot\langle \va,\vx_{\perp}\rangle=\cos \theta\cdot a_1+\sin\theta\cdot a_2.
		\]
Given the orthogonality of $a_1$ and $a_2$, as evidenced by $\mathbb{E}(a_1\cdot a_2)=0$, we can conclude that $a_1$ and $a_2$ are independently sampled from $\mathcal{N}(0,1)$.

For (\ref{eqn: E1}):  
				\begin{eqnarray*}
			\begin{split}
				&\mathbb{E}\left(\langle \boldsymbol{a}, \boldsymbol{y} \rangle^2 \cdot\boldsymbol{1}_{\{\langle \boldsymbol{a}, \boldsymbol{x} \rangle \geq 0\}}\right) = \mathbb{E}\left((\cos\theta \cdot a_1 +\sin\theta \cdot a_2)^2\cdot\boldsymbol{1}_{\{a_1 \geq 0\}}\right)\\
				 =& \frac{1}{2\pi}\int_{-\infty}^{+\infty}\int_{0}^{+\infty}(\cos\theta \cdot t_1 + \sin\theta \cdot t_2)^2\cdot \exp\Big(-\frac{t_1^2 + t_2^2}{2}\Big)dt_1dt_2\\			
				 	 \overset{(a)}=& \frac{1}{2\pi}\int_{0}^{+\infty}\int_{-\frac{\pi}{2}}^{\frac{\pi}{2}} \rho^3\cdot  (\cos\theta \cdot \cos\varphi+ \sin\theta \cdot \sin\varphi)^2\cdot \exp\Big(-\frac{\rho^2}{2}\Big) d\varphi d\rho\\
				 =&\frac{1}{2\pi}\int_{0}^{+\infty}\rho^3\cdot \exp\left(-\frac{\rho^2}{2}\right)d\rho\cdot \int_{-\frac{\pi}{2}}^{\frac{\pi}{2}}\cos^2(\varphi-\theta)d\varphi= \frac{1}{2}.
			\end{split}
		\end{eqnarray*}
		The equality $(a)$ is derived through polar coordinate transformation by setting $t_1 = \rho\cdot \cos\varphi$ and $t_2 = \rho\cdot \sin\varphi$, with $\varphi \in [-\frac{\pi}{2}, \frac{\pi}{2}]$ and $\rho \in [0, \infty)$.
		
For (\ref{eqn: E2}): 
		\[
			\begin{split}
						&\mathbb{E}\left(\langle \boldsymbol{a}, \boldsymbol{y} \rangle^2 \cdot \boldsymbol{1}_{\{\alpha\geq \langle \boldsymbol{a}, \boldsymbol{x} \rangle > 0\}}\right)=\mathbb{E}\left((\cos\theta \cdot a_1 +\sin\theta \cdot a_2)^2\cdot \boldsymbol{1}_{\{\alpha\geq a_1 > 0\}}\right)\\
				=& \frac{1}{2\pi}\int_{-\infty}^{\infty}\int_{0}^{\alpha}(\cos\theta \cdot t_1 + \sin\theta \cdot t_2)^2\cdot \exp\Big(-\frac{t_1^2 + t_2^2}{2}\Big)dt_1dt_2 \\
				\leq &\frac{1}{2\pi}\int_{-\infty}^{\infty}\int_{0}^{\alpha}( |t_1| + |t_2|)^2\cdot \exp\Big(-\frac{t_1^2 + t_2^2}{2}\Big)dt_1dt_2\\
				\leq &\frac{1}{2\pi}\int_{-\infty}^{\infty}\int_{0}^{\alpha}t_1^2\cdot \exp\Big(-\frac{t_1^2 + t_2^2}{2}\Big)dt_1dt_2+\frac{1}{2\pi}\int_{-\infty}^{\infty}\int_{0}^{\alpha}t_2^2\cdot \exp\Big(-\frac{t_1^2 + t_2^2}{2}\Big)dt_1dt_2\\
				&+\frac{1}{\pi}\int_{-\infty}^{\infty}\int_{0}^{\alpha}|t_1\cdot t_2| \cdot \exp\Big(-\frac{t_1^2 + t_2^2}{2}\Big)dt_1dt_2\\
				\overset{(a)}\leq & \frac{1}{\sqrt{2\pi}}\alpha+\frac{1}{\sqrt{2\pi}}\alpha+\frac{\sqrt{2}}{\sqrt{\pi}}\alpha\leq 2\alpha,
			\end{split}
		\]
		where (a) follows from
		\[\int_{0}^{\alpha}t_1^2\cdot\exp(-t_1^2/2)dt_1\leq \int_{0}^{\alpha}t_1\cdot\exp(-t_1^2/2)dt_1\leq \int_{0}^{\alpha}\exp(-t_1^2/2)dt_1\leq \alpha,\] 
		$\frac{1}{\sqrt{2\pi}}\int_{-\infty}^{\infty}|t_2|\cdot\exp(-t_2^2/2)dt_2\leq 1$, and
		\[
		\frac{1}{\sqrt{2\pi}}\int_{-\infty}^{\infty}\exp(-t_2^2/2)dt_2=\frac{1}{\sqrt{2\pi}}\int_{-\infty}^{\infty}t_2^2\cdot \exp(-t_2^2/2)dt_2=1.\]
		These conclude the proofs of both (\ref{eqn: E1}) and (\ref{eqn: E2}).
		\end{proof}

		\begin{lemma}\label{lem: beta}
		Let $\boldsymbol{a}\in \mathbb{R}^{n}$ be a random vector  such that $\boldsymbol{a} \sim \mathcal{N}(\boldsymbol{0},\boldsymbol{I}_n)$. For any fixed unit vector $\boldsymbol{y}\in \mathbb{S}^{n-1}$ and  constant $\beta\geq 10$, the following inequality holds:
		\begin{equation}
				\mathbb{E}\left(\langle \boldsymbol{a}, \boldsymbol{y} \rangle^2\cdot  \boldsymbol{1}_{\{|\langle \boldsymbol{a}, \boldsymbol{y} \rangle| \geq \beta\}}\right) \leq \exp(-{\beta^2}/{4}).
		\end{equation}
		\end{lemma}
		\begin{proof}
		Take $a=\langle \boldsymbol{a},\boldsymbol{y}\rangle$. Then $a\sim \mathcal{N}(0,1)$ and 
		\[
			\begin{split}
				\mathbb{E}\left(\langle \boldsymbol{a}, \boldsymbol{y} \rangle^2\cdot \boldsymbol{1}_{\{|\langle \boldsymbol{a}, \boldsymbol{y} \rangle| \geq \beta\}}\right)
				 =& \mathbb{E}\left(a^2\cdot\boldsymbol{1}_{\{|a| \geq \beta\}}\right)=\frac{2}{\sqrt{2\pi}}\int_{\beta}^{\infty}t^2\cdot\exp(-t^2/2)dt\\				\leq &\frac{\sqrt{2}}{\sqrt{\pi}}\int_{\beta}^{\infty} t^3\cdot\exp(-t^2/2)dt
				= \frac{\sqrt{2}}{\sqrt{\pi}}\cdot 2(\beta^2/2+1)\cdot\exp(-\beta^2/2)\\
				\leq& (\beta^2+2)\cdot \exp(-\beta^2/2)
				\leq \exp(-\beta^2/4),
			\end{split}
		\]
		The last line follows from $\beta\geq 10$,  which leads to the conclusion $\beta^2+2\leq \exp(\beta^2/4)$.		\end{proof}

 		\begin{lemma}\label{lem: final}
			Let $\boldsymbol{a}\in \mathbb{R}^{n}$ be a random vector  such that $\boldsymbol{a} \sim \mathcal{N}(\boldsymbol{0},\boldsymbol{I}_n)$. For any fixed unit vectors $\boldsymbol{x}, \boldsymbol{y} \in \mathbb{S}^{n-1}$ and positive constants $\alpha,\beta$ satisfying $\alpha\in (0,1)$ and $\beta\geq 10$, the following inequalities hold:
		\begin{equation}
		\label{eqn: alpha_beta}
		\begin{split}
				\mathbb{E}\left(\langle \boldsymbol{a}, \boldsymbol{y} \rangle^2\cdot \boldsymbol{1}_{\{\langle \boldsymbol{a}, \boldsymbol{x} \rangle > \alpha \}}\cdot \boldsymbol{1}_{\{|\langle \boldsymbol{a}, \boldsymbol{y} \rangle| < \beta\}}\right) \geq   \frac{1}{2}- 2\alpha - \exp(-\beta^2/4),
				\end{split}
		\end{equation}
		and
				\begin{equation}
				\label{eqn: alpha_new}
				\mathbb{E}\left(\langle \boldsymbol{a}, \boldsymbol{y} \rangle^2 \cdot\boldsymbol{1}_{\{\langle \boldsymbol{a}, \boldsymbol{x} \rangle \geq - \alpha \}}\cdot \boldsymbol{1}_{\{|\langle \boldsymbol{a}, \boldsymbol{y} \rangle| < \beta\}}\right)\leq \mathbb{E}\left(\langle \boldsymbol{a}, \boldsymbol{y} \rangle^2\cdot \boldsymbol{1}_{\{\langle \boldsymbol{a}, \boldsymbol{x} \rangle \geq -\alpha \}}\right)\leq \frac{1}{2}+2\alpha.
		\end{equation}
			\end{lemma}
			\begin{proof}
			Applying Lemma \ref{lem: mid} and Lemma \ref{lem: beta}, we can directly derive the following inequalities:
			\[
		\begin{aligned}
				&\mathbb{E}\left(\langle \boldsymbol{a}, \boldsymbol{y} \rangle^2\cdot \boldsymbol{1}_{\{\langle \boldsymbol{a}, \boldsymbol{x} \rangle > \alpha \}}\cdot \boldsymbol{1}_{\{|\langle \boldsymbol{a}, \boldsymbol{y} \rangle| < \beta\}}\right) \\
				\geq & \mathbb{E}\left(\langle \boldsymbol{a}, \boldsymbol{y} \rangle^2\cdot  \boldsymbol{1}_{\{\langle \boldsymbol{a}, \boldsymbol{x} \rangle \geq 0\}}\right)-\mathbb{E}\left(\langle \boldsymbol{a}, \boldsymbol{y} \rangle^2\cdot \boldsymbol{1}_{\{0<\langle \boldsymbol{a}, \boldsymbol{x} \rangle \leq\alpha\} }\right)-\mathbb{E}\left(\langle \boldsymbol{a}, \boldsymbol{y} \rangle^2\cdot \boldsymbol{1}_{\{|\langle \boldsymbol{a}, \boldsymbol{y} \rangle|\geq  \beta \}}\right) \\
				\geq &  \frac{1}{2}- 2\alpha - \exp(-\beta^2/4),
				\end{aligned}
		\]
		and
				\[
				\begin{aligned}
				\mathbb{E}\left(\langle \boldsymbol{a}, \boldsymbol{y} \rangle^2\cdot \boldsymbol{1}_{\{\langle \boldsymbol{a}, \boldsymbol{x} \rangle \geq - \alpha \}}\cdot \boldsymbol{1}_{\{|\langle \boldsymbol{a}, \boldsymbol{y} \rangle| < \beta\}} \right)\leq &\mathbb{E}\left(\langle \boldsymbol{a}, \boldsymbol{y} \rangle^2\cdot \boldsymbol{1}_{\{\langle \boldsymbol{a}, \boldsymbol{x} \rangle \geq -\alpha \}}\right)\\
				\leq &\mathbb{E}\left(\langle \boldsymbol{a}, \boldsymbol{y} \rangle^2\cdot \boldsymbol{1}_{\{0>\langle \boldsymbol{a}, \boldsymbol{x} \rangle \geq -\alpha \}}\right)+\mathbb{E}\left(\langle \boldsymbol{a}, \boldsymbol{y} \rangle^2\cdot \boldsymbol{1}_{\{\langle \boldsymbol{a}, \boldsymbol{x} \rangle \geq 0 \}}\right) \\
				\leq &\frac{1}{2}+2\alpha.
				\end{aligned}
		\]
	These inequalities establish the desired results, thereby completing the proof.
			\end{proof}

	\subsection{Technical Concentration Inequalities}
Moreover, we introduce a series of critical concentration lemmas that are instrumental in proving our main theorem. While these lemmas share similar technical foundations, their specific details and applications differ significantly. For the sake of completeness and to elucidate these nuances, we present the full proofs of these lemmas in this section.
	\begin{lemma}\label{lem: middle}
	Let $\{\va_i\}_{i=1}^m$ be a sequence of independent standard Gaussian random vectors in $\mathbb{R}^n$, i.e., $\va_i \sim \mathcal{N}(\boldsymbol{0}, \boldsymbol{I}_n)$,  for $i = 1, \ldots, m$.  For any fixed positive constants $\delta$ and $\beta$ such that $\delta<1$ and $\beta\geq 10$, with probability at least $1-2\exp(-c_\delta\cdot m)$, we have:
		    \[
	    	\frac{1}{m}\cdot \sum_{i=1}^{m}(\langle \boldsymbol{a}_{i}, \boldsymbol{z} \rangle)^2\cdot \boldsymbol{1}_{\{|\langle \boldsymbol{a}_{i}, \boldsymbol{z} \rangle| \geq \beta\}}
	    		 \leq \exp(-0.9^2\beta^2/4) +\delta,
	    \]
	    for any $\boldsymbol{z}\in S\cap \mathbb{S}^{n-1}$, provided that $m\gtrsim \delta^{-4}\cdot\omega^2((S-S)\cap \mathbb{B}^{n})$. Here $c_\delta$ is a positive constant depending on $\delta$. 
	\end{lemma}		
	
	\begin{proof}
	Denote the function $g_\beta: \mathbb{R}\rightarrow \mathbb{R}$ as 
	    \begin{eqnarray*}
	    	\begin{split}
	    		g_{\beta}(t) := 
	    		\begin{cases}
	    			1, &t \geq \beta;\\
	    			\frac{(t-0.9\beta)^2}{(0.1\beta)^2}, &0.9\beta \leq t < \beta; \\
	    			0, &t < 0.9\beta.
	    		\end{cases}
	    	\end{split}
	    \end{eqnarray*}
	    By direct calculation, we have 
	    \begin{equation}\label{X1_temp}
	    \begin{aligned}
	    \frac{1}{m}\cdot \sum_{i=1}^{m}(\langle \boldsymbol{a}_{i}, \boldsymbol{z} \rangle)^2 \cdot \boldsymbol{1}_{\{|\langle \boldsymbol{a}_{i}, \boldsymbol{z} \rangle| \geq \beta\}}\leq& \frac{1}{m}\cdot \sum_{i=1}^{m}(\langle \boldsymbol{a}_{i}, \boldsymbol{z} \rangle)^2\cdot g_{\beta}(|\langle \boldsymbol{a}_{i}, \boldsymbol{z} \rangle|)\\
	    \leq &\frac{1}{m}\cdot \sum_{i=1}^{m}(\langle \boldsymbol{a}_{i}, \boldsymbol{z} \rangle)^2\cdot \boldsymbol{1}_{\{|\langle \boldsymbol{a}_{i}, \boldsymbol{z} \rangle| \geq 0.9\beta\}}.
	    \end{aligned}
	    \end{equation}
	    Therefore, for any $\vz\in S\cap \mathbb{S}^{n-1}$,  in order to get the upper bound of 
	    \[
	     \frac{1}{m}\cdot \sum_{i=1}^{m}(\langle \boldsymbol{a}_{i}, \boldsymbol{z} \rangle)^2 \cdot \boldsymbol{1}_{\{|\langle \boldsymbol{a}_{i}, \boldsymbol{z} \rangle| \geq \beta\}},
	     \]
	      it is enough to estimate the upper bound of 
	    \[
	    \frac{1}{m}\cdot \sum_{i=1}^{m}(\langle \boldsymbol{a}_{i}, \boldsymbol{z} \rangle)^2\cdot  g_{\beta}(|\langle \boldsymbol{a}_{i}, \boldsymbol{z} \rangle|).
	    \] 
	    
	  For any fixed $\delta_0<1/45$,  take $\mathcal{N}(S\cap \mathbb{S}^{n-1},\delta_0)$ as an $\delta_0$-net of $S\cap \mathbb{S}^{n-1}$. According to Lemma \ref{lem: cover_union} and Lemma \ref{lem: sudakov}, with probability at least $1-2\exp(-c_{\delta_0}\cdot m)$, we have:
	    \begin{equation}
	    \label{eqn: upper_x1_temp}
	    \begin{aligned}
	    \frac{1}{m}\cdot \sum_{i=1}^{m}(\langle \boldsymbol{a}_{i}, \boldsymbol{z}_{\delta_0} \rangle)^2 \cdot g_{\beta}(|\langle \boldsymbol{a}_{i}, \boldsymbol{z}_{\delta_0} \rangle|) \leq & \frac{1}{m}\cdot \sum_{i=1}^{m}\mathbb{E}(\langle \boldsymbol{a}_{i}, \boldsymbol{z}_{\delta_0} \rangle)^2 \cdot g_{\beta}(|\langle \boldsymbol{a}_{i}, \boldsymbol{z}_{\delta_0} \rangle|) +   \delta_0\\
	    \overset{(a)}\leq & \exp(-0.9^2\beta^2/4)+\delta_0,
	    \end{aligned}
	    \end{equation}
	  for all $\boldsymbol{z}_{\delta_0}\in\mathcal{N}(S\cap \mathbb{S}^{n-1},\delta_0)$, provided that 
	  \[
	  m\gtrsim \delta_0^{-4}\cdot \omega^2((S-S)\cap \mathbb{B}^{n})\geq \delta_0^{-4}\cdot\omega^2(S\cap \mathbb{S}^{n-1}).
	  \]
	   Here $c_{\delta_0}$ is a positive constant depending on $\delta_0$. $(a)$ is based on (\ref{X1_temp}) and Lemma \ref{lem: beta}.
	  
	  Take $h:\mathbb{R}\rightarrow \mathbb{R}$ such that $h_{\beta}(t) = \sqrt{g_{\beta}(t)}$. We claim that, for all $t_1,t_2\in \mathbb{R}$, 
	  \begin{equation}\label{eqn: claim_x1}
	  \left|(t_1)^2 \cdot g_{\beta}(|t_1|)- (t_2)^2 \cdot g_{\beta}(|t_2|) \right|\leq 11\cdot |t_1-t_2|\cdot \big(|t_1|+|t_2|\big).
	  \end{equation}
	   Denote $\vA\in \mathbb{R}^{m\times n}$ as $\vA:=(\va_1,\ldots,\va_m)^{T}$. Then  for any $\boldsymbol{z}\in S\cap \mathbb{S}^{n-1}$ with $\boldsymbol{z}_{\delta_0}\in \mathcal{N}(S\cap \mathbb{S}^{n-1},\delta_0)$ such that $\|\boldsymbol{z}-\boldsymbol{z}_{\delta_0}\|_2 \leq \delta_0$, we have:
	   \begin{equation}
	   \small
	   \label{eqn: upper_tempXX2}
	   \begin{split}
	    		&\frac{1}{m}\cdot \sum_{i=1}^{m}(\langle \boldsymbol{a}_{i}, \boldsymbol{z} \rangle)^2 \cdot g_{\beta}\big(|\langle \boldsymbol{a}_{i}, \boldsymbol{z} \rangle|\big)\\
	    		 \leq& \frac{1}{m}\cdot \sum_{i=1}^{m}(\langle \boldsymbol{a}_{i}, \boldsymbol{z}_{\delta_0} \rangle)^2 \cdot g_{\beta}\big(|\langle \boldsymbol{a}_{i}, \boldsymbol{z}_{\delta_0} \rangle|\big) + \frac{1}{m}\cdot \sum_{i=1}^{m}\Big|(\langle \boldsymbol{a}_{i}, \boldsymbol{z} \rangle)^2 \cdot g_{\beta}\big(|\langle \boldsymbol{a}_{i}, \boldsymbol{z} \rangle|\big) - (\langle \boldsymbol{a}_{i}, \boldsymbol{z}_{\delta_0} \rangle)^2 \cdot g_{\beta}\big(|\langle \boldsymbol{a}_{i}, \boldsymbol{z}_{\delta_0} \rangle|\big) \Big|\\
	    		 			\overset{(b)} \leq& \exp(-0.9^2\beta^2/4)+\delta_0+ \frac{11}{m}\cdot \sum_{i=1}^{m}\left|\langle \boldsymbol{a}_{i}, \boldsymbol{z} - \boldsymbol{z}_{\delta_0} \rangle \right|\cdot \big( \left|\langle \boldsymbol{a}_{i}, \boldsymbol{z} \rangle \right| + \left| \langle \boldsymbol{a}_{i}, \boldsymbol{z}_{\delta_0} \rangle \right|\big)\\
	    		 \leq& \exp(-0.9^2\beta^2/4)+\delta_0 + \frac{11}{m}\cdot \|\vA(\vz-\vz_{\delta_0})\|_2\cdot \|\vA \vz\|_2+\frac{11}{m}\cdot \|\vA(\vz-\vz_{\delta_0})\|_2\cdot \|\vA \vz_{\delta_0}\|_2\\
			 \overset{(c)}\leq& \exp(-0.9^2\beta^2/4)+\delta_0 +22\cdot (1+\delta_0)\cdot \|\boldsymbol{z} - \boldsymbol{z}_{\delta_0}\|_2	 \leq \exp(-0.9^2\beta^2/4)+45\delta_0.
	    \end{split}
	    \end{equation}
Here $(b)$ is based on  (\ref{eqn: upper_x1_temp}) and (\ref{eqn: claim_x1}), and $(c)$ is according to  Lemma \ref{lem: RIP} such that $\frac{1}{\sqrt{m}}\cdot \|\vA(\vz-\vz_{\delta_0})\|_2\leq \sqrt{1+\delta_0}\cdot\|\vz-\vz_{\delta_0}\|_2$, $\frac{1}{\sqrt{m}}\cdot \|\vA\vz\|_2\leq \sqrt{1+\delta_0}\|\vz\|_2\leq 1+\delta_0$, and $\frac{1}{\sqrt{m}}\cdot \|\vA\vz_{\delta}\|_2\leq \sqrt{1+\delta_0}\|\vz_{\delta}\|_2\leq 1+\delta_0$. Plugging (\ref{eqn: upper_tempXX2}) into (\ref{X1_temp}) with $\delta=45\delta_0<1$, we can conclude that:
\[
\frac{1}{m}\cdot \sum_{i=1}^{m}(\langle \boldsymbol{a}_{i}, \boldsymbol{z} \rangle)^2\cdot  \boldsymbol{1}_{\{|\langle \boldsymbol{a}_{i}, \boldsymbol{z} \rangle| \geq \beta\}}\leq \exp(-0.9^2\beta^2/4)+\delta.	
\]

 The only thing left is to prove (\ref{eqn: claim_x1}).  For all $t_1,t_2\in \mathbb{R}$, we can directly have  
	  \[
	  \left|t_1 \cdot h_{\beta}(t_1)- t_2 \cdot h_{\beta}(t_2) \right| \leq 11\cdot |t_1-t_2|.
	  \]
Combined with the fact that $|h_{\beta}(t)|\leq 1$ for all $t\in \mathbb{R}$, it leads to 
\begin{equation}
\label{eqn: h_temp}
\begin{aligned}
 \left|(t_1)^2 \cdot g_{\beta}(|t_1|)- (t_2)^2 \cdot g_{\beta}(|t_2|) \right|=&\Big(|t_1|\cdot h_\beta(|t_1|)-|t_2|\cdot h_{\beta}(|t_2|)\Big)\cdot \Big(|t_1|\cdot h_\beta(|t_1|)+|t_2|\cdot h_{\beta}(|t_2|)\Big)\\
 \leq & 11\cdot \big(|t_1|-|t_2|\big)\cdot \big(|t_1|+|t_2|\big)\leq 11\cdot |t_1-t_2|\cdot \big(|t_1|+|t_2|\big),
 \end{aligned}
\end{equation}
and the proof of (\ref{eqn: claim_x1}) is completed. 
	  	\end{proof}
	
	\begin{lemma}\label{lem: upper}
	Let $\{\va_i\}_{i=1}^m$ be a sequence of independent standard Gaussian random vectors in $\mathbb{R}^n$, i.e., $\va_i \sim \mathcal{N}(\boldsymbol{0}, \boldsymbol{I}_n)$,  for $i = 1, \ldots, m$.  For any fixed  positive constants $\alpha,\beta,\delta$ such that $\alpha,\delta\in (0,1)$ and $\beta\geq 10$,  with probability at least $1-2\exp(-c_{\alpha,\beta,\delta} \cdot m)$, we have 
	    \[
	     \frac{1}{m}\cdot \sum_{i=1}^{m}(\langle \boldsymbol{a}_{i}, \boldsymbol{z} \rangle)^2 \cdot \boldsymbol{1}_{\{\langle \boldsymbol{a}_{i}, \boldsymbol{x} \rangle \geq -\alpha\} }\cdot \boldsymbol{1}_{\{|\langle \boldsymbol{a}_{i}, \boldsymbol{z} \rangle| < \beta\}}\leq  \frac{1}{2}+2.2\alpha+\frac{2}{\beta^2}+\delta
	    \]
	    for any $\boldsymbol{x},\boldsymbol{z}\in S\cap \mathbb{S}^{n-1}$, provided that 
	    \[
	    m\gtrsim \frac{\beta^8}{\alpha^4\cdot \delta^4}\cdot \omega^2((S-S)\cap \mathbb{B}^{n}).
	    \]
	      Here $c_{\alpha,\beta,\delta}$ is a positive constant depending on $\alpha,\beta$ and $\delta$. 
	    \end{lemma}
	    \begin{proof}
	    Denote the function $g_{\alpha}:\mathbb{R}\rightarrow \mathbb{R}$ as 
	    	    \begin{equation}
		    \label{eqn: g_alpha1}
	    	\begin{split}
	    		g_{\alpha}(t) =
	    		\begin{cases}
	    			1, &t \geq -\alpha;\\
	    			\frac{(t+1.1\alpha)^2}{(0.1\alpha)^2}, &-1.1\alpha \leq t < -\alpha; \\
	    			0,&t < -1.1\alpha.
	    		\end{cases}
	    	\end{split}
	    \end{equation}
	    By direct calculation, we have 
  \[
	    	\begin{split}
	    		& \frac{1}{m}\cdot \sum_{i=1}^{m}(\langle \boldsymbol{a}_{i}, \boldsymbol{z} \rangle)^2 \cdot \boldsymbol{1}_{\{\langle \boldsymbol{a}_{i}, \boldsymbol{x} \rangle \geq -\alpha\} }\cdot \boldsymbol{1}_{\{|\langle \boldsymbol{a}_{i}, \boldsymbol{z} \rangle| < \beta\}}\\
	    		 \leq &\frac{1}{m}\cdot \sum_{i=1}^{m}(\langle \boldsymbol{a}_{i}, \boldsymbol{z} \rangle)^2 \cdot  g_{\alpha}(\langle \boldsymbol{a}_{i}, \boldsymbol{x} \rangle)\cdot  \boldsymbol{1}_{\{|\langle \boldsymbol{a}_{i}, \boldsymbol{z} \rangle| < \beta\}}\\
			 \leq &\frac{1}{m}\cdot \sum_{i=1}^{m}(\langle \boldsymbol{a}_{i}, \boldsymbol{z} \rangle)^2\cdot  \boldsymbol{1}_{\{\langle \boldsymbol{a}_{i}, \boldsymbol{x} \rangle \geq -1.1\alpha\} }\cdot \boldsymbol{1}_{\{|\langle \boldsymbol{a}_{i}, \boldsymbol{z} \rangle| < \beta\}}.
	    	\end{split}
	    \]
	      Therefore, in order to get the upper bound of 
	      \[
	        \frac{1}{m}\cdot \sum_{i=1}^{m}(\langle \boldsymbol{a}_{i}, \boldsymbol{z} \rangle)^2 \cdot \boldsymbol{1}_{\{\langle \boldsymbol{a}_{i}, \boldsymbol{x} \rangle \geq -\alpha\} }\cdot \boldsymbol{1}_{\{|\langle \boldsymbol{a}_{i}, \boldsymbol{z} \rangle| < \beta\}},
	        \]
	         it is enough to estimate the upper bound of 
	    \[
	    \frac{1}{m}\cdot \sum_{i=1}^{m}(\langle \boldsymbol{a}_{i}, \boldsymbol{z} \rangle)^2 \cdot  g_{\alpha}(\langle \boldsymbol{a}_{i}, \boldsymbol{x} \rangle)\cdot  \boldsymbol{1}_{\{|\langle \boldsymbol{a}_{i}, \boldsymbol{z} \rangle| < \beta\}}.
	    \] 
	    
	    For any fixed $\delta_0<1/2$, set $\epsilon:=\frac{\alpha\cdot \delta_0}{80\beta^2}$. Then we have 
	    $
	    \epsilon\leq \delta_0.
	    $ Take $\mathcal{N}(S\cap \mathbb{S}^{n-1},\epsilon)$ as an $\epsilon$-net of $S\cap \mathbb{S}^{n-1}$. We claim that with probability at least $1-2\exp(-c_{\alpha,\beta,\delta_0} \cdot m)$ for some positive constant $c_{\alpha,\beta,\delta_0}$ depending on $\alpha,\beta$, and $\delta_0$, the following inequalities hold:
	    \begin{equation}
	    \label{eqn: temp_x2_1}
	      \frac{1}{m}\cdot \sum_{i=1}^{m}(\langle \boldsymbol{a}_{i}, \boldsymbol{z}_{\epsilon} \rangle)^2\cdot  g_{\alpha}(\langle \boldsymbol{a}_{i}, \boldsymbol{x}_{\epsilon} \rangle) \cdot \boldsymbol{1}_{\{|\langle \boldsymbol{a}_{i}, \boldsymbol{z}_{\epsilon} \rangle| < \beta^2\}} \leq  \frac{1}{2}+2.2\alpha+\delta_0,
	    \end{equation}
	    \begin{equation}
	    \label{eqn: temp_x2_2}
	     \frac{1}{m}\cdot \sum_{i=1}^{m}(\langle \boldsymbol{a}_{i}, \boldsymbol{z} \rangle)^2 \cdot g_{\alpha}(\langle \boldsymbol{a}_{i}, \boldsymbol{x} \rangle)\cdot \boldsymbol{1}_{\{|\langle \boldsymbol{a}_{i}, \boldsymbol{z} \rangle| < \beta\}}\cdot \boldsymbol{1}_{\{ |\langle \boldsymbol{a}_{i}, \boldsymbol{z}_{\epsilon} \rangle| \geq \beta^2\}}\leq \frac{2}{\beta^2},
	    \end{equation}
	    \begin{equation}
	    \label{eqn: temp_x2_3}
	     \frac{1}{m}\cdot \sum_{i=1}^{m}\left|(\langle \boldsymbol{a}_{i}, \boldsymbol{z} \rangle)^2  \cdot g_{\alpha}(\langle \boldsymbol{a}_{i}, \boldsymbol{x} \rangle) - (\langle \boldsymbol{a}_{i}, \boldsymbol{z}_{\epsilon} \rangle)^2 \cdot  g_{\alpha}(\langle \boldsymbol{a}_{i}, \boldsymbol{x}_{\epsilon} \rangle)\right|\cdot \boldsymbol{1}_{\{|\langle \boldsymbol{a}_{i}, \boldsymbol{z}_{\epsilon} \rangle| < \beta^2\}}\cdot \boldsymbol{1}_{\{|\langle \boldsymbol{a}_{i}, \boldsymbol{z} \rangle| < \beta\}}\leq \delta_0,
	    \end{equation}  
	    for any $\vx,\vz\in S\cap \mathbb{S}^{n-1}$ and $\vx_{\epsilon},\vz_{\epsilon}\in \mathcal{N}(S\cap \mathbb{S}^{n-1},\epsilon)$ such that $\|\vx-\vx_{\epsilon}\|_2\leq \epsilon$ and $\|\vz-\vz_{\epsilon}\|_2\leq \epsilon$, provided that 
	    \[
	    m\gtrsim \frac{\beta^8}{\alpha^4\cdot \delta_0^4}\cdot\omega^2((S-S)\cap \mathbb{B}^{n}).
	    \] 
	The proofs of them will be presented at the end.
	
	Therefore, for all $\vx,\vz\in S\cap \mathbb{S}^{n-1}$ and $\vx_{\epsilon},\vz_{\epsilon}\in \mathcal{N}(S\cap \mathbb{S}^{n-1},\epsilon)$ such that $\|\vx-\vx_{\epsilon}\|_2\leq \epsilon$ and $\|\vz-\vz_{\epsilon}\|_2\leq \epsilon$, we have 
	\begin{equation}\label{eqn: phi1_temp1}
	    	\begin{split}
	    		& \frac{1}{m}\cdot \sum_{i=1}^{m}(\langle \boldsymbol{a}_{i}, \boldsymbol{z} \rangle)^2 \cdot \boldsymbol{1}_{\{\langle \boldsymbol{a}_{i}, \boldsymbol{x} \rangle \geq -\alpha\} }\cdot \boldsymbol{1}_{\{|\langle \boldsymbol{a}_{i}, \boldsymbol{z} \rangle| < \beta\}}\\
			\leq&  \frac{1}{m}\cdot \sum_{i=1}^{m}(\langle \boldsymbol{a}_{i}, \boldsymbol{z} \rangle)^2 \cdot  g_{\alpha}(\langle \boldsymbol{a}_{i}, \boldsymbol{x} \rangle)\cdot\boldsymbol{1}_{\{|\langle \boldsymbol{a}_{i}, \boldsymbol{z} \rangle| < \beta\}}\\
	    		 = &\frac{1}{m}\cdot \sum_{i=1}^{m}(\langle \boldsymbol{a}_{i}, \boldsymbol{z}_{\epsilon} \rangle)^2 \cdot g_{\alpha}(\langle \boldsymbol{a}_{i}, \boldsymbol{x}_{\epsilon} \rangle)\cdot\boldsymbol{1}_{\{|\langle \boldsymbol{a}_{i}, \boldsymbol{z}_{\epsilon} \rangle| < \beta^2\}}
	    		 + \frac{1}{m}\cdot \sum_{i=1}^{m}(\langle \boldsymbol{a}_{i}, \boldsymbol{z} \rangle)^2 \cdot g_{\alpha}(\langle \boldsymbol{a}_{i}, \boldsymbol{x} \rangle)\cdot \boldsymbol{1}_{\{|\langle \boldsymbol{a}_{i}, \boldsymbol{z} \rangle| < \beta\}} \\
	    		&- \frac{1}{m}\cdot \sum_{i=1}^{m}(\langle \boldsymbol{a}_{i}, \boldsymbol{z}_{\epsilon} \rangle)^2\cdot  g_{\alpha}(\langle \boldsymbol{a}_{i}, \boldsymbol{x}_{\epsilon} \rangle)\cdot\boldsymbol{1}_{\{|\langle \boldsymbol{a}_{i}, \boldsymbol{z}_{\epsilon} \rangle| < \beta^2\}}\\
	    		 \leq& \frac{1}{m}\cdot \sum_{i=1}^{m}(\langle \boldsymbol{a}_{i}, \boldsymbol{z}_{\epsilon} \rangle)^2  \cdot g_{\alpha}(\langle \boldsymbol{a}_{i}, \boldsymbol{x}_{\epsilon} \rangle)\cdot \boldsymbol{1}_{\{|\langle \boldsymbol{a}_{i}, \boldsymbol{z}_{\epsilon} \rangle| < \beta^2\}}\\
	    		  &+ \frac{1}{m}\cdot\sum_{i=1}^{m}(\langle \boldsymbol{a}_{i}, \boldsymbol{z} \rangle)^2\cdot  g_{\alpha}(\langle \boldsymbol{a}_{i}, \boldsymbol{x} \rangle)\cdot\boldsymbol{1}_{\{|\langle \boldsymbol{a}_{i}, \boldsymbol{z} \rangle| < \beta\}}\cdot \boldsymbol{1}_{\{ |\langle \boldsymbol{a}_{i}, \boldsymbol{z}_{\epsilon} \rangle| \geq \beta^2\}} \\
	    		 \quad&   +  \frac{1}{m}\sum_{i=1}^{m}\left|(\langle \boldsymbol{a}_{i}, \boldsymbol{z} \rangle)^2  \cdot g_{\alpha}(\langle \boldsymbol{a}_{i}, \boldsymbol{x} \rangle) - (\langle \boldsymbol{a}_{i}, \boldsymbol{z}_{\epsilon} \rangle)^2 \cdot  g_{\alpha}(\langle \boldsymbol{a}_{i}, \boldsymbol{x}_{\epsilon} \rangle)\right|\cdot \boldsymbol{1}_{\{|\langle \boldsymbol{a}_{i}, \boldsymbol{z}_{\epsilon} \rangle| < \beta^2\}}\cdot \boldsymbol{1}_{\{|\langle \boldsymbol{a}_{i}, \boldsymbol{z} \rangle| < \beta\}}\\
	    		\leq &\frac{1}{2}+2.2\alpha+\frac{2}{\beta^2}+2\delta_0=\frac{1}{2}+2.2\alpha+\frac{2}{\beta^2}+\delta,	 
	    	\end{split}
\end{equation}
where $\delta:=2\delta_0$. Then the proof is completed. 
	
	Now we prove the results in (\ref{eqn: temp_x2_1}), (\ref{eqn: temp_x2_2}) and (\ref{eqn: temp_x2_3}). 
	   
	    According to Lemma \ref{lem: cover_union} and Lemma \ref{lem: sudakov}, with probability at least $1-2\exp(-c_{\alpha,\beta,\delta_0}\cdot m)$, we directly have:
	    \begin{equation}
	    \label{eqn: upper_x2_temp}
	    \begin{aligned}
	 &  \frac{1}{m}\cdot\sum_{i=1}^{m}(\langle \boldsymbol{a}_{i}, \boldsymbol{z}_{\epsilon} \rangle)^2\cdot  g_{\alpha}(\langle \boldsymbol{a}_{i}, \boldsymbol{x}_{\epsilon} \rangle) \cdot \boldsymbol{1}_{\{|\langle \boldsymbol{a}_{i}, \boldsymbol{z}_{\epsilon} \rangle| < \beta^2\}} \\
	  \leq&  \frac{1}{m}\cdot\sum_{i=1}^{m}\mathbb{E}\langle \boldsymbol{a}_{i}, \boldsymbol{z}_{\epsilon}\rangle^2  \cdot g_{\alpha}(\langle \boldsymbol{a}_{i}, \boldsymbol{x}_{\epsilon} \rangle)\cdot  \boldsymbol{1}_{\{|\langle \boldsymbol{a}_{i}, \boldsymbol{z}_{\epsilon} \rangle| < \beta^2\}}+   \delta_0\\
	   \overset{(a)} \leq & \frac{1}{m}\cdot \sum_{i=1}^{m}\mathbb{E}\langle \boldsymbol{a}_{i}, \boldsymbol{z}_{\epsilon}\rangle^2  \cdot \boldsymbol{1}_{\{\langle \boldsymbol{a}_{i}, \boldsymbol{x}_{\epsilon}\rangle \geq -1.1\alpha \}} +   \delta_0 \\
	    \overset{(b)}\leq & \frac{1}{2}+2.2\alpha+\delta_0,
	    \end{aligned}
	    \end{equation}
	  for any $\boldsymbol{z}_{\epsilon},\boldsymbol{x}_{\epsilon}\in \mathcal{N}(S\cap \mathbb{S}^{n-1},\epsilon)$,  provided that 
	  \[
	  m\gtrsim \frac{1}{\epsilon^4}\cdot \omega^2((S-S)\cap \mathbb{B}^{n})\geq \frac{1}{\epsilon^2\cdot\delta_0^2}\cdot \omega^2((S-S)\cap \mathbb{B}^{n}) \geq   \frac{1}{\epsilon^2\cdot\delta_0^2}\cdot \omega^2(S\cap \mathbb{S}^{n-1}).
	  \] Here $(a)$ if follows from the definition of $g_{\alpha}$ in (\ref{eqn: g_alpha1}), and $(b)$ is based on (\ref{eqn: alpha_new}) in Lemma \ref{lem: final} by replacing $\alpha$ into $1.1\alpha$. Then the proof of (\ref{eqn: temp_x2_1}) is completed. 
	    
	    Besides, (\ref{eqn: temp_x2_2}) can be directly derived by:
\[
\begin{aligned}
 &\frac{1}{m}\cdot \sum_{i=1}^{m}(\langle \boldsymbol{a}_{i}, \boldsymbol{z} \rangle)^2 \cdot g_{\alpha}(\langle \boldsymbol{a}_{i}, \boldsymbol{x} \rangle)\cdot\boldsymbol{1}_{\{|\langle \boldsymbol{a}_{i}, \boldsymbol{z} \rangle| < \beta\}}\cdot \boldsymbol{1}_{\{|\langle \boldsymbol{a}_{i}, \boldsymbol{z}_{\epsilon} \rangle| \geq \beta^2\}}\\
 \leq&  \frac{1}{m}\cdot \sum_{i=1}^{m}\beta^2\cdot \boldsymbol{1}_{\{|\langle \boldsymbol{a}_{i}, \boldsymbol{z}_{\epsilon} \rangle| \geq \beta^2\}}\leq \frac{2}{\beta^2}.
 \end{aligned}
\]
The last inequality above follows from 
\[
\#\{i\in\{1,\cdots,m\}\ :\ |\langle \boldsymbol{a}_{i}, \boldsymbol{z}_{\epsilon} \rangle| > \beta^2\}\leq \frac{2}{\beta^4}.
\]
If the cardinality of the set $\{i\in\{1,\ldots,m\} : |\langle \boldsymbol{a}_{i}, \boldsymbol{z}_{\epsilon} \rangle| > \beta^2\}$  exceeds $\frac{2}{\beta^4}$, then it would follow that $\frac{1}{m}\cdot \sum_{i=1}^{m}\langle \boldsymbol{a}_i,\boldsymbol{z}_{\epsilon}\rangle^2 > 2$. However, this contradicts the fact that $\frac{1}{m}\cdot \sum_{i=1}^{m}\langle \boldsymbol{a}_i,\boldsymbol{z}_{\epsilon}\rangle^2 \leq 1+\delta_0 \leq 2$ for all $\boldsymbol{z}_{\epsilon} \in \mathcal{N}(S\cap \mathbb{S}^{n-1},\epsilon)$, as established by Lemma \ref{lem: cover_union}.

Moreover, denote $ h_{\alpha}(t) = \sqrt{ g_{\alpha}(t)}$.  Then it obtains that 
   \begin{equation}\label{eqn: phi1_temp2}
	   	\begin{split}
	   		& \frac{1}{m}\cdot \sum_{i=1}^{m}\left|(\langle \boldsymbol{a}_{i}, \boldsymbol{z} \rangle)^2\cdot g_{\alpha}(\langle \boldsymbol{a}_{i}, \boldsymbol{x} \rangle) - (\langle \boldsymbol{a}_{i}, \boldsymbol{z}_{\epsilon} \rangle)^2 \cdot  g_{\alpha}(\langle \boldsymbol{a}_{i}, \boldsymbol{x}_{\epsilon} \rangle)\right| \cdot \boldsymbol{1}_{\{|\langle \boldsymbol{a}_{i}, \boldsymbol{z}_{\epsilon} \rangle| < \beta^2\}}\cdot \boldsymbol{1}_{\{ |\langle \boldsymbol{a}_{i}, \boldsymbol{z} \rangle| < \beta\}}\\
	   		 \leq&  \frac{\beta + \beta^2}{m}\cdot\sum_{i=1}^{m} \left|\langle \boldsymbol{a}_{i}, \boldsymbol{z} \rangle\cdot { h_{\alpha}(\langle \boldsymbol{a}_{i}, \boldsymbol{x} \rangle)} - \langle\boldsymbol{a}_{i}, \boldsymbol{z}_{\epsilon} \rangle\cdot { h_{\alpha}(\langle \boldsymbol{a}_{i}, \boldsymbol{x}_{\epsilon} \rangle)}\right|\\
	   		 \leq&  \frac{\beta + \beta^2}{m}\cdot\sum_{i=1}^{m}\Big(\left| \langle \boldsymbol{a}_{i}, \boldsymbol{z} -\vz_{\epsilon}\rangle \cdot { h_{\alpha}(\langle \boldsymbol{a}_{i}, \boldsymbol{x} \rangle)} \right| 
+ \left|\langle \boldsymbol{a}_{i}, \boldsymbol{z}_{\epsilon} \rangle|\cdot| { h_{\alpha}(\langle \boldsymbol{a}_{i}, \boldsymbol{x} \rangle)}- { h_{\alpha}(\langle \boldsymbol{a}_{i}, \boldsymbol{x}_{\epsilon} \rangle)} \right|\Big)\\
	   		 \leq&  \frac{\beta + \beta^2}{m}\cdot\sum_{i=1}^{m}\Big(\left| \langle \boldsymbol{a}_{i}, \boldsymbol{z} -\vz_{\epsilon}\rangle  \right| 
+ \left|\langle \boldsymbol{a}_{i}, \boldsymbol{z}_{\epsilon} \rangle|\cdot| { h_{\alpha}(\langle \boldsymbol{a}_{i}, \boldsymbol{x} \rangle)}- { h_{\alpha}(\langle \boldsymbol{a}_{i}, \boldsymbol{x}_{\epsilon} \rangle)} \right|\Big)\\
			\overset{(c)}{\leq} &  \frac{\beta + \beta^2}{m}\cdot\sum_{i=1}^{m}(\left| \langle \boldsymbol{a}_{i}, \boldsymbol{z}-\boldsymbol{z}_{\epsilon} \rangle \right| + \frac{10}{\alpha}\cdot |\langle \boldsymbol{a}_{i}, \boldsymbol{z}_{\epsilon} \rangle| \cdot \left| \langle \boldsymbol{a}_{i}, \boldsymbol{x}-\boldsymbol{x}_{\epsilon} \rangle \right|)\\
			\overset{(d)}\leq &(\beta+\beta^2)\cdot \frac{10}{\alpha}\cdot(1+\delta_0)\cdot (\|\boldsymbol{z}-\boldsymbol{z}_{\epsilon}\|_2+\|\boldsymbol{x}-\boldsymbol{x}_{\epsilon}\|_2)\leq \frac{80\beta^2\cdot\epsilon}{\alpha}=\delta_0,
	   		\end{split}
	   \end{equation}
where (c) follows from $\left| h_{\alpha}(t_1)-  h_{\alpha}(t_2) \right| \leq \frac{10}{\alpha}\cdot |t_1-t_2|$, for any $t_1,t_2\in \mathbb{R}$, and (d) is according to Lemma \ref{lem: RIP}. Thus the proof of (\ref{eqn: temp_x2_3}) is completed. 
	    \end{proof}
	    
	    \begin{lemma}\label{lem: lower}
	Let $\{\va_i\}_{i=1}^m$ be a sequence of independent standard Gaussian random vectors in $\mathbb{R}^n$, i.e., $\va_i \sim \mathcal{N}(\boldsymbol{0}, \boldsymbol{I}_n)$,  for $i = 1, \ldots, m$.  For any fixed  non-negative constants $\alpha,\beta,\delta$ such that $\alpha,\delta\in (0,1)$ and $\beta\geq 10$,  with probability at least $1-2\exp(-c_{\alpha,\beta,\delta}\cdot m)$, we have: 
	    \[
	     \frac{1}{m}\cdot \sum_{i=1}^{m}(\langle \boldsymbol{a}_{i}, \boldsymbol{z} \rangle)^2\cdot \boldsymbol{1}_{\{\langle \boldsymbol{a}_{i}, \boldsymbol{x} \rangle >\alpha\}}\cdot \boldsymbol{1}_{\{|\langle \boldsymbol{a}_{i}, \boldsymbol{z} \rangle| < \beta\}}\geq  \frac{1}{2}-2.2\alpha -\exp(-\beta/4)- \frac{2}{\beta}-\delta,
	    \]
	    for any $\boldsymbol{x},\boldsymbol{z}\in S\cap \mathbb{S}^{n-1}$, provided that 
	    \[
	    m\gtrsim \frac{\beta^8}{\alpha^4\cdot \delta^4}\cdot \omega^2((S-S)\cap \mathbb{B}^{n}).
	    \]     Here $c_{\alpha,\beta_,\delta}$ is a positive constant depending on $\alpha,\beta$ and $\delta$. 
	    \end{lemma}
\begin{proof}
Denote the function $g_{\alpha}:\mathbb{R}\rightarrow \mathbb{R}$ as:
\begin{equation}
\label{eqn: g_alpha_2}
	    		g_{\alpha}(t): =
	    		\begin{cases}
	    			1, &t \geq 1.1\alpha;\\
	    			\frac{(t-\alpha)^2}{(0.1\alpha)^2}, & \alpha \leq t < 1.1\alpha; \\
	    			0,&t < \alpha.
	    		\end{cases}
	    \end{equation}
	    By direct calculation, we have 
  \[
	    	\begin{split}
	    		& \frac{1}{m}\cdot \sum_{i=1}^{m}(\langle \boldsymbol{a}_{i}, \boldsymbol{z} \rangle)^2 \cdot \boldsymbol{1}_{\{\langle \boldsymbol{a}_{i}, \boldsymbol{x} \rangle >\alpha\} }\cdot \boldsymbol{1}_{\{|\langle \boldsymbol{a}_{i}, \boldsymbol{z} \rangle| < \beta\}}\\
	    		 \geq &\frac{1}{m}\cdot \sum_{i=1}^{m}(\langle \boldsymbol{a}_{i}, \boldsymbol{z} \rangle)^2 \cdot  g_{\alpha}(\langle \boldsymbol{a}_{i}, \boldsymbol{x} \rangle)\cdot  \boldsymbol{1}_{\{|\langle \boldsymbol{a}_{i}, \boldsymbol{z} \rangle| < \beta\}}\\
			 \geq &\frac{1}{m}\cdot \sum_{i=1}^{m}(\langle \boldsymbol{a}_{i}, \boldsymbol{z} \rangle)^2\cdot  \boldsymbol{1}_{\{\langle \boldsymbol{a}_{i}, \boldsymbol{x} \rangle \geq 1.1\alpha\} }\cdot \boldsymbol{1}_{\{|\langle \boldsymbol{a}_{i}, \boldsymbol{z} \rangle| < \beta\}}.
	    	\end{split}
	    \]

    Therefore, in order to get the lower bound of 
    \[
      \frac{1}{m}\cdot \sum_{i=1}^{m}(\langle \boldsymbol{a}_{i}, \boldsymbol{z} \rangle)^2 \cdot \boldsymbol{1}_{\{\langle \boldsymbol{a}_{i}, \boldsymbol{x} \rangle > \alpha\} }\cdot \boldsymbol{1}_{\{|\langle \boldsymbol{a}_{i}, \boldsymbol{z} \rangle| < \beta\}},
      \]
       it is enough to estimate the lower bound of 
	    \[
	    \frac{1}{m}\cdot \sum_{i=1}^{m}(\langle \boldsymbol{a}_{i}, \boldsymbol{z} \rangle)^2 \cdot  g_{\alpha}(\langle \boldsymbol{a}_{i}, \boldsymbol{x} \rangle)\cdot  \boldsymbol{1}_{\{|\langle \boldsymbol{a}_{i}, \boldsymbol{z} \rangle| < \beta\}}.
	    \] 
	    
	    For any fixed $\delta_0<1/2$, set $\epsilon:=\frac{\alpha\cdot \delta_0}{80\beta^2}$. Then $\epsilon\leq \delta_0$. Take $\mathcal{N}(S\cap \mathbb{S}^{n-1},\epsilon)$ as an $\epsilon$-net of $S\cap \mathbb{S}^{n-1}$. We claim that with probability at least $1-2\exp(-c_{\alpha,\beta,\delta_0} \cdot m)$ for some positive constant $c_{\alpha,\beta,\delta_0}$ depending on $\alpha,\beta$ and $\delta_0$, the following inequalities hold:
	    \begin{equation}
	    \label{eqn: temp_x3_1}
	      \frac{1}{m}\cdot \sum_{i=1}^{m}(\langle \boldsymbol{a}_{i}, \boldsymbol{z}_{\epsilon} \rangle)^2\cdot  g_{\alpha}(\langle \boldsymbol{a}_{i}, \boldsymbol{x}_{\epsilon} \rangle) \cdot \boldsymbol{1}_{\{|\langle \boldsymbol{a}_{i}, \boldsymbol{z}_{\epsilon} \rangle| < \sqrt{\beta}\}} \geq  \frac{1}{2}-2.2\alpha-\exp(-\beta/4)-\delta_0,
	    \end{equation}
	    \begin{equation}
	    \label{eqn: temp_x3_2}
	     \frac{1}{m}\cdot \sum_{i=1}^{m}(\langle \boldsymbol{a}_{i}, \boldsymbol{z}_{\epsilon} \rangle)^2 \cdot g_{\alpha}(\langle \boldsymbol{a}_{i}, \boldsymbol{x}_{\epsilon} \rangle)\cdot \boldsymbol{1}_{\{|\langle \boldsymbol{a}_{i}, \boldsymbol{z} \rangle| \geq \beta\}}\cdot \boldsymbol{1}_{\{ |\langle \boldsymbol{a}_{i}, \boldsymbol{z}_{\epsilon} \rangle| < \sqrt{\beta}\}}\leq \frac{2}{\beta},
	    \end{equation}
	    \begin{equation}
	    \label{eqn: temp_x3_3}
	     \frac{1}{m}\cdot \sum_{i=1}^{m}\left|(\langle \boldsymbol{a}_{i}, \boldsymbol{z} \rangle)^2  \cdot g_{\alpha}(\langle \boldsymbol{a}_{i}, \boldsymbol{x} \rangle) - (\langle \boldsymbol{a}_{i}, \boldsymbol{z}_{\epsilon} \rangle)^2 \cdot  g_{\alpha}(\langle \boldsymbol{a}_{i}, \boldsymbol{x}_{\epsilon} \rangle)\right|\cdot \boldsymbol{1}_{\{|\langle \boldsymbol{a}_{i}, \boldsymbol{z}_{\epsilon} \rangle| < \sqrt{\beta}\}}\cdot \boldsymbol{1}_{\{|\langle \boldsymbol{a}_{i}, \boldsymbol{z} \rangle| < \beta\}}\leq \delta_0,
	    \end{equation}  
	    for any $\vx,\vz\in S\cap \mathbb{S}^{n-1}$ and $\vx_{\epsilon},\vz_{\epsilon}\in \mathcal{N}(S\cap \mathbb{S}^{n-1},\epsilon)$ such that $\|\vx-\vx_{\epsilon}\|_2\leq \epsilon$ and $\|\vz-\vz_{\epsilon}\|_2\leq \epsilon$, provided that 
	    \[
	    m\gtrsim \frac{\beta^8}{\alpha^4\cdot \delta_0^4}\cdot \omega^2((S-S)\cap \mathbb{B}^{n}).
	    \] 
	The proofs of them will be presented at the end.

Therefore, for all $\vx,\vz\in S\cap \mathbb{S}^{n-1}$ and $\vx_{\epsilon},\vz_{\epsilon}\in \mathcal{N}(S\cap \mathbb{S}^{n-1},\epsilon)$ such that $\|\vx-\vx_{\epsilon}\|_2\leq \epsilon$ and $\|\vz-\vz_{\epsilon}\|_2\leq \epsilon$, we have 

 \begin{equation}\label{eqn: phi2_temp1}
 \small
	    	\begin{split}
	    		&\frac{1}{m}\cdot \sum_{i=1}^{m}(\langle \boldsymbol{a}_{i}, \boldsymbol{z} \rangle)^2 \cdot \boldsymbol{1}_{\{\langle \boldsymbol{a}_{i}, \boldsymbol{x} \rangle >\alpha\} }\cdot \boldsymbol{1}_{\{|\langle \boldsymbol{a}_{i}, \boldsymbol{z} \rangle| < \beta\}}\leq  \frac{1}{m}\cdot \sum_{i=1}^{m}(\langle \boldsymbol{a}_{i}, \boldsymbol{z} \rangle)^2\cdot   g_{\alpha}(\langle \boldsymbol{a}_{i}, \boldsymbol{x} \rangle)\cdot\boldsymbol{1}_{\{|\langle \boldsymbol{a}_{i}, \boldsymbol{z} \rangle| < {\beta}\}}\\
	    		 =&  \frac{1}{m}\cdot \sum_{i=1}^{m}(\langle \boldsymbol{a}_{i}, \boldsymbol{z} \rangle)^2\cdot    g_{\alpha}(\langle \boldsymbol{a}_{i}, \boldsymbol{x} \rangle)\cdot \boldsymbol{1}_{\{|\langle \boldsymbol{a}_{i}, \boldsymbol{z} \rangle| < \beta\}} - \frac{1}{m}\cdot \sum_{i=1}^{m}(\langle \boldsymbol{a}_{i}, \boldsymbol{z}_{\epsilon} \rangle)^2 \cdot  g_{\alpha}(\langle \boldsymbol{a}_{i}, \boldsymbol{x}_{\epsilon} \rangle)\cdot \boldsymbol{1}_{\{|\langle \boldsymbol{a}_{i}, \boldsymbol{z}_{\epsilon} \rangle| < \sqrt{\beta}\}}\\
			 &+\frac{1}{m}\cdot \sum_{i=1}^{m}(\langle \boldsymbol{a}_{i}, \boldsymbol{z}_{\epsilon} \rangle)^2  \cdot  g_{\alpha}(\langle \boldsymbol{a}_{i}, \boldsymbol{x}_{\epsilon} \rangle)\cdot \boldsymbol{1}_{\{|\langle \boldsymbol{a}_{i}, \boldsymbol{z}_{\epsilon} \rangle| < \sqrt{\beta}\}}\\
	    		 \geq& \frac{1}{m}\cdot \sum_{i=1}^{m}(\langle \boldsymbol{a}_{i}, \boldsymbol{z}_{\epsilon} \rangle)^2 \cdot  g_{\alpha}(\langle \boldsymbol{a}_{i}, \boldsymbol{x}_{\epsilon} \rangle)\cdot \boldsymbol{1}_{\{|\langle \boldsymbol{a}_{i}, \boldsymbol{z}_{\epsilon} \rangle| < \sqrt{\beta}\}} \\
			 &- \frac{1}{m}\cdot \sum_{i=1}^{m}(\langle \boldsymbol{a}_{i}, \boldsymbol{z}_{\epsilon} \rangle)^2 \cdot  g_{\alpha}(\langle \boldsymbol{a}_{i}, \boldsymbol{x}_{\epsilon} \rangle)\cdot \boldsymbol{1}_{\{|\langle \boldsymbol{a}_{i}, \boldsymbol{z} \rangle| \geq \beta\}}\cdot \boldsymbol{1}_{\{|\langle \boldsymbol{a}_{i}, \boldsymbol{z}_{\epsilon} \rangle| < \sqrt{\beta}\}} \\
	    		 & - \frac{1}{m}\cdot \sum_{i=1}^{m}\Big|(\langle \boldsymbol{a}_{i}, \boldsymbol{z} \rangle)^2 \cdot  g_{\alpha}(\langle \boldsymbol{a}_{i}, \boldsymbol{x} \rangle) - (\langle \boldsymbol{a}_{i}, \boldsymbol{z}_{\epsilon} \rangle)^2 \cdot  g_{\alpha}(\langle \boldsymbol{a}_{i}, \boldsymbol{x}_{\epsilon} \rangle)\Big|\cdot \boldsymbol{1}_{\{|\langle \boldsymbol{a}_{i}, \boldsymbol{z}_{\epsilon} \rangle| < \sqrt{\beta}\}}\cdot \boldsymbol{1}_{\{ |\langle \boldsymbol{a}_{i}, \boldsymbol{z} \rangle| < \beta\}}\\
			{\geq} &\frac{1}{2}-2.2\alpha -\exp(-\beta/4)- \frac{2}{\beta}-2\delta_0=\frac{1}{2}-2.2\alpha -\exp(-\beta/4)- \frac{2}{\beta}-\delta,
	    	\end{split}
	    \end{equation}
	    where $\delta:=2\delta_0$. Then the proof is completed. 
	
	Now we prove the results in (\ref{eqn: temp_x3_1}), (\ref{eqn: temp_x3_2}) and (\ref{eqn: temp_x3_3}). 
	
	 According to Lemma \ref{lem: cover_union} and Lemma \ref{lem: sudakov}, with probability at least $1-2\exp(-c_{\alpha,\beta,\delta_0} \cdot m)$, we directly have:
	    \begin{equation}
	    \begin{aligned}
	  & \frac{1}{m}\cdot \sum_{i=1}^{m}(\langle \boldsymbol{a}_{i}, \boldsymbol{z}_{\epsilon} \rangle)^2\cdot  g_{\alpha}(\langle \boldsymbol{a}_{i}, \boldsymbol{x}_{\epsilon} \rangle) \cdot \boldsymbol{1}_{\{|\langle \boldsymbol{a}_{i}, \boldsymbol{z}_{\epsilon} \rangle| < \sqrt{\beta}\}} \\
	  \geq&  \frac{1}{m}\cdot \sum_{i=1}^{m}\mathbb{E}\langle \boldsymbol{a}_{i}, \boldsymbol{z}_{\epsilon}\rangle^2  \cdot g_{\alpha}(\langle \boldsymbol{a}_{i}, \boldsymbol{x}_{\epsilon} \rangle)\cdot  \boldsymbol{1}_{\{|\langle \boldsymbol{a}_{i}, \boldsymbol{z}_{\epsilon} \rangle| < \sqrt{\beta}\}}-\delta_0\\
	   \overset{(a)} \geq & \frac{1}{m}\cdot \sum_{i=1}^{m}\mathbb{E}(\langle \boldsymbol{a}_{i}, \boldsymbol{z}_{\epsilon} \rangle)^2\cdot  \boldsymbol{1}_{\{\langle \boldsymbol{a}_{i}, \boldsymbol{x}_{\epsilon} \rangle \geq 1.1\alpha\} }\cdot \boldsymbol{1}_{\{|\langle \boldsymbol{a}_{i}, \boldsymbol{z}_{\epsilon} \rangle| < \sqrt{\beta}\}} -  \delta_0 \\
	    \overset{(b)}\leq & \frac{1}{2}-2.2\alpha -\exp(-\beta/4)-\delta_0,
	    \end{aligned}
	    \end{equation}
	  for any $\boldsymbol{z}_{\epsilon},\boldsymbol{x}_{\epsilon}\in \mathcal{N}(S\cap \mathbb{S}^{n-1},\epsilon)$,  provided that 
	  \[
	   m\gtrsim \frac{1}{\epsilon^4}\cdot \omega^2((S-S)\cap \mathbb{B}^{n})\geq \frac{1}{\epsilon^2\cdot\delta_0^2}\cdot \omega^2((S-S)\cap \mathbb{B}^{n}) \geq   \frac{1}{\epsilon^2\cdot\delta_0^2}\cdot \omega^2(S\cap \mathbb{S}^{n-1}).
	  \] Here $(a)$ if follows from the definition of $g_{\alpha}$ in (\ref{eqn: g_alpha_2}), and $(b)$ is based on (\ref{eqn: alpha_beta}) in Lemma \ref{lem: final} by replacing $\alpha$ into $1.1\alpha$. Then the proof of (\ref{eqn: temp_x3_1}) is completed. 

Besides, (\ref{eqn: temp_x3_2}) can be directly derived by:
\[
\begin{aligned}
& \frac{1}{m}\cdot \sum_{i=1}^{m}(\langle \boldsymbol{a}_{i}, \boldsymbol{z}_{\epsilon} \rangle)^2 \cdot g_{\alpha}(\langle \boldsymbol{a}_{i}, \boldsymbol{x}_{\epsilon} \rangle)\cdot\boldsymbol{1}_{\{|\langle \boldsymbol{a}_{i}, \boldsymbol{z} \rangle| \geq \beta\}}\cdot \boldsymbol{1}_{\{|\langle \boldsymbol{a}_{i}, \boldsymbol{z}_{\epsilon} \rangle| < \sqrt{\beta}\}}\\
 \leq&  \frac{1}{m}\cdot \sum_{i=1}^{m}\beta\cdot \boldsymbol{1}_{\{|\langle \boldsymbol{a}_{i}, \boldsymbol{z} \rangle| \geq \beta\}}\leq \frac{2}{\beta}.
 \end{aligned}
\]
The last inequality above follows from 
\[
\#\{i\in\{1,\cdots,m\}\ :\ |\langle \boldsymbol{a}_{i}, \boldsymbol{z} \rangle| > \beta\}\leq \frac{2}{\beta^2}.
\]
If the cardinality of the set $\{i\in\{1,\ldots,m\} : |\langle \boldsymbol{a}_{i}, \boldsymbol{z} \rangle| > \beta\}$ exceeds $\frac{2}{\beta^2}$, then it would follow that $\frac{1}{m}\cdot \sum_{i=1}^{m}(\langle \boldsymbol{a}_i,\boldsymbol{z}\rangle)^2 > 2$. However, this contradicts the fact that $\frac{1}{m}\cdot \sum_{i=1}^{m}(\langle \boldsymbol{a}_i,\boldsymbol{z}\rangle)^2 \leq 1+\delta_0 \leq 2$  for all $\boldsymbol{z} \in S\cap \mathbb{S}^{n-1}$, as established by Lemma \ref{lem: RIP}.

Moreover, denote $ h_{\alpha}(t) = \sqrt{ g_{\alpha}(t)}$.  Then it obtains that:
\begin{equation}\label{eqn: phi2_temp2}
	    	\begin{split}
& \frac{1}{m}\cdot \sum_{i=1}^{m}\left|(\langle \boldsymbol{a}_{i}, \boldsymbol{z} \rangle)^2 \cdot  g_{\alpha}(\langle \boldsymbol{a}_{i}, \boldsymbol{x} \rangle) - (\langle \boldsymbol{a}_{i}, \boldsymbol{z}_{\epsilon} \rangle)^2 \cdot  g_{\alpha}(\langle \boldsymbol{a}_{i}, \boldsymbol{x}_{\epsilon} \rangle)\right| \cdot \boldsymbol{1}_{\{|\langle \boldsymbol{a}_{i}, \boldsymbol{z}_{\epsilon} \rangle| <  \sqrt{\beta}\}}\cdot \boldsymbol{1}_{\{|\langle \boldsymbol{a}_{i}, \boldsymbol{z} \rangle| < \beta\}}\\
	    		 \leq&  \frac{\beta+\sqrt{\beta} }{m}\cdot \sum_{i=1}^{m}\left|\langle \boldsymbol{a}_{i}, \boldsymbol{z} \rangle \cdot { h_{\alpha}(\langle \boldsymbol{a}_{i}, \boldsymbol{x} \rangle)} - \langle \boldsymbol{a}_{i}, \boldsymbol{z}_{\epsilon} \rangle \cdot { h_{\alpha}(\langle \boldsymbol{a}_{i}, \boldsymbol{x}_{\epsilon} \rangle)}\right|\\
			  \leq&  \frac{\beta + \sqrt{\beta}}{m}\cdot\sum_{i=1}^{m}\Big(\left| \langle \boldsymbol{a}_{i}, \boldsymbol{z} -\vz_{\epsilon}\rangle \cdot { h_{\alpha}(\langle \boldsymbol{a}_{i}, \boldsymbol{x} \rangle)} \right| 
+ \left|\langle \boldsymbol{a}_{i}, \boldsymbol{z}_{\epsilon} \rangle|\cdot| { h_{\alpha}(\langle \boldsymbol{a}_{i}, \boldsymbol{x} \rangle)}- { h_{\alpha}(\langle \boldsymbol{a}_{i}, \boldsymbol{x}_{\epsilon} \rangle)} \right|\Big)\\
			  \leq&  \frac{\beta + \sqrt{\beta}}{m}\cdot\sum_{i=1}^{m}\Big(\left| \langle \boldsymbol{a}_{i}, \boldsymbol{z} -\vz_{\epsilon}\rangle  \right| 
+ \left|\langle \boldsymbol{a}_{i}, \boldsymbol{z}_{\epsilon} \rangle|\cdot| { h_{\alpha}(\langle \boldsymbol{a}_{i}, \boldsymbol{x} \rangle)}- { h_{\alpha}(\langle \boldsymbol{a}_{i}, \boldsymbol{x}_{\epsilon} \rangle)} \right|\Big)\\
			\overset{(c)}{\leq} &  \frac{\beta + \sqrt{\beta}}{m}\cdot\sum_{i=1}^{m}\Big(\left| \langle \boldsymbol{a}_{i}, \boldsymbol{z}-\boldsymbol{z}_{\epsilon} \rangle \right| + \frac{10}{\alpha}\cdot |\langle \boldsymbol{a}_{i}, \boldsymbol{z}_{\epsilon} \rangle| \cdot \left| \langle \boldsymbol{a}_{i}, \boldsymbol{x}-\boldsymbol{x}_{\epsilon} \rangle \right|\Big)\\
			\overset{(d)}\leq &(\beta+\sqrt{\beta})\cdot \frac{10}{\alpha}\cdot(1+\delta_0)\cdot (\|\boldsymbol{z}-\boldsymbol{z}_{\epsilon}\|_2+\|\boldsymbol{x}-\boldsymbol{x}_{\epsilon}\|_2)\leq \frac{80\beta\cdot\epsilon}{\alpha}\leq \delta_0,
	    	\end{split}
\end{equation}
where (c) follows from $\left| h_{\alpha}(t_1)-  h_{\alpha}(t_2) \right| \leq \frac{10}{\alpha}\cdot |t_1-t_2|$, for any $t_1,t_2\in \mathbb{R}$, and (d) is according to Lemma \ref{lem: RIP}. Thus the proof of (\ref{eqn: temp_x3_3}) is completed. 

\end{proof}

\subsection{Proof of Theorem \ref{thm: small}}
\begin{proof}[Proof of Theorem \ref{thm: small}]
If $\boldsymbol{x}=\boldsymbol{y}$, the inequality (\ref{eqn: small_C}) is satisfied immediately. Now we discuss the case when $\boldsymbol{x}\neq \boldsymbol{y}$. Similarly as the proof of Theorem \ref{thm: large}, we need only consider the scenario where $\boldsymbol{x}\in S\cap \mathbb{S}^{n-1}$ and $\boldsymbol{y}\in S\cap\mathbb{B}^{n}$, subject to the condition $\|\boldsymbol{x}-\boldsymbol{y}\|_2\leq C$ and $\vx\neq \vy$. 

Denote $\boldsymbol{a}_i\in \mathbb{R}^n$, $i=1,\ldots,m$, as the row elements of the matrix $\boldsymbol{A}$. Then we have 
\begin{equation}
\label{eqn: sigma_xy}
\frac{1}{m}\cdot \frac{\|\sigma(\boldsymbol{A}\boldsymbol{x})-\sigma(\boldsymbol{A}\boldsymbol{y})\|_2^2}{\|\boldsymbol{x}-\boldsymbol{y}\|_2^2}=\frac{1}{m}\cdot\sum_{i=1}^m\frac{(\sigma(\langle \boldsymbol{a}_i,\boldsymbol{x}\rangle)-\sigma(\langle \boldsymbol{a}_i,\boldsymbol{y}\rangle))^2}{\|\boldsymbol{x}-\boldsymbol{y}\|_2^2}.
\end{equation}
Take $\boldsymbol{z}:=\frac{\boldsymbol{x}-\boldsymbol{y}}{\|\boldsymbol{x}-\boldsymbol{y}\|_2}$. For any $i\in\{1,\ldots,m\}$, if $|\langle \boldsymbol{a}_i,\boldsymbol{x}\rangle |> \alpha$ and $|\langle \boldsymbol{a}_i,\boldsymbol{z}\rangle|\leq \beta$, we have 
\begin{equation}
\label{eqn: sign}
\text{sign}(\langle \boldsymbol{a}_i,\boldsymbol{x}\rangle)=\text{sign}(\langle \boldsymbol{a}_i,\|\vx-\vy\|_2\cdot \vz+\vy\rangle)=\text{sign}(\langle \boldsymbol{a}_i,\boldsymbol{y}\rangle),
\end{equation}
as 
\[
|\langle\va_i,\vx\rangle|> \alpha>C\cdot \beta\geq \|\vx-\vy\|_2\cdot|\langle \va_i,\vz\rangle| =|\langle \va_i,\|\vx-\vy\|_2\cdot\vz\rangle|.
\] 
Thus, the lower bound of (\ref{eqn: sigma_xy}) becomes:
\begin{equation}\label{eqn: ind_temp1}
\begin{split}
&\frac{1}{m}\cdot\sum_{i=1}^m\frac{(\sigma(\langle \boldsymbol{a}_i,\boldsymbol{x}\rangle)-\sigma(\langle \boldsymbol{a}_i,\boldsymbol{y}\rangle))^2}{\|\boldsymbol{x}-\boldsymbol{y}\|_2^2}\\
\geq&\frac{1}{m}\cdot \sum_{i=1}^m \frac{(\sigma(\langle \boldsymbol{a}_i,\boldsymbol{x}\rangle)-\sigma(\langle \boldsymbol{a}_i,\boldsymbol{y}\rangle))^2\cdot \boldsymbol{1}_{\{\langle \boldsymbol{a}_{i}, \boldsymbol{x} \rangle > \alpha\}} \cdot \boldsymbol{1}_{\{|\langle \boldsymbol{a}_{i}, \boldsymbol{{z}} \rangle| \leq  \beta \}}}{\|\boldsymbol{x}-\boldsymbol{y}\|_2^2}\\
\overset{(a)}=&\frac{1}{m}\cdot \sum_{i=1}^m \frac{(\langle \boldsymbol{a}_i,\boldsymbol{x}\rangle-\langle \boldsymbol{a}_i,\boldsymbol{y}\rangle)^2\cdot \boldsymbol{1}_{\{\langle \boldsymbol{a}_{i}, \boldsymbol{x} \rangle > \alpha\}} \cdot \boldsymbol{1}_{\{|\langle \boldsymbol{a}_{i}, \boldsymbol{{z}} \rangle| \leq  \beta \}}}{\|\boldsymbol{x}-\boldsymbol{y}\|_2^2}\\
=&\frac{1}{m}\cdot \sum_{i=1}^m(\langle \boldsymbol{a}_{i}, \boldsymbol{z}\rangle)^2\cdot \boldsymbol{1}_{\{\langle \boldsymbol{a}_{i}, \boldsymbol{x} \rangle > \alpha\}}\cdot \boldsymbol{1}_{\{ |\langle \boldsymbol{a}_{i}, \boldsymbol{{z}} \rangle| \leq\beta\} }\\
\overset{(b)}\geq  &\frac{1}{2}-2.2\alpha -\exp(-\beta/4)- \frac{2}{\beta}-\delta\\
\overset{(c)}\geq & \frac{1}{2}-3\alpha- \frac{3}{\beta}-2\delta.
\end{split}
\end{equation}
Here $(a)$ follows from (\ref{eqn: sign}). $(b)$ is based on Lemma  \ref{lem: lower}. $(c)$ relies on $\beta\geq 10$, which leads to $\exp(-\beta/4)\leq 1/\beta$.

On the other hand, the upper bound of (\ref{eqn: sigma_xy}) becomes:

\begin{equation}\label{eqn: ind_temp2}
\begin{split}
&\frac{1}{m}\cdot \sum_{i=1}^m\frac{(\sigma(\langle \boldsymbol{a}_i,\boldsymbol{x}\rangle)-\sigma(\langle \boldsymbol{a}_i,\boldsymbol{y}\rangle))^2}{\|\boldsymbol{x}-\boldsymbol{y}\|_2^2}\\
\leq &\frac{1}{m}\cdot \sum_{i=1}^m\frac{(\sigma(\langle \boldsymbol{a}_i,\boldsymbol{x}\rangle)-\sigma(\langle \boldsymbol{a}_i,\boldsymbol{y}\rangle))^2\cdot \boldsymbol{1}_{\{\langle \va_i,\vx\rangle<-\alpha\}}\cdot \boldsymbol{1}_{\{|\langle \va_i,\vz\rangle|<\beta\}}}{\|\boldsymbol{x}-\boldsymbol{y}\|_2^2}\\
&+\frac{1}{m}\cdot \sum_{i=1}^m\frac{(\sigma(\langle \boldsymbol{a}_i,\boldsymbol{x}\rangle)-\sigma(\langle \boldsymbol{a}_i,\boldsymbol{y}\rangle))^2\cdot \boldsymbol{1}_{\{\langle \va_i,\vx\rangle\geq -\alpha\}}\cdot \boldsymbol{1}_{\{|\langle \va_i,\vz\rangle|<\beta\}}}{\|\boldsymbol{x}-\boldsymbol{y}\|_2^2}\\
&+\frac{1}{m}\cdot \sum_{i=1}^m\frac{(\sigma(\langle \boldsymbol{a}_i,\boldsymbol{x}\rangle)-\sigma(\langle \boldsymbol{a}_i,\boldsymbol{y}\rangle))^2\cdot  \boldsymbol{1}_{\{|\langle \va_i,\vz\rangle|\geq \beta\}}}{\|\boldsymbol{x}-\boldsymbol{y}\|_2^2}\\
\overset{(a)}\leq &\frac{1}{m}\cdot \sum_{i=1}^{m}(\langle \boldsymbol{a}_{i}, \boldsymbol{{z}} \rangle)^2\cdot  \boldsymbol{1}_{\{\langle \boldsymbol{a}_{i}, \boldsymbol{x} \rangle \geq -\alpha\}}\cdot \boldsymbol{1}_{\{ |\langle \boldsymbol{a}_{i}, \boldsymbol{{z}} \rangle| < \beta\}} + \frac{1}{m}\sum_{i=1}^{m}\langle (\boldsymbol{a}_{i}, \boldsymbol{{z}} \rangle)^2\cdot  \boldsymbol{1}_{\{|\langle \boldsymbol{a}_{i}, \boldsymbol{{z}} \rangle| \geq \beta\}}\\
\overset{(b)}\leq &\frac{1}{2}+2.2\alpha+\frac{2}{\beta^2}+\exp(-0.9^2\beta^2/4) +2\delta\\
\overset{(c)}\leq & \frac{1}{2}+3\alpha+\frac{3}{\beta}+2\delta.
 \end{split}
\end{equation}
Here $(a)$ follows from (\ref{eqn: sign}), which leads to 
\[
(\sigma(\langle \boldsymbol{a}_i,\boldsymbol{x}\rangle)-\sigma(\langle \boldsymbol{a}_i,\boldsymbol{y}\rangle))^2\cdot \boldsymbol{1}_{\{\langle \va_i,\vx\rangle\geq -\alpha\}}\cdot \boldsymbol{1}_{\{|\langle \va_i,\vz\rangle|<\beta\}}=0,
\]
 for $i=1,\ldots,m$, and $|\sigma(t_1)-\sigma(t_2)|\leq |t_1-t_2|$ for all $t_1,t_2\in \mathbb{R}$. $(b)$ is based on  Lemma \ref{lem: middle} and Lemma \ref{lem: upper}. $(c)$ relies on the fact that $\beta\geq 10$, which leads to $\exp(-0.9^2\beta^2/4)\leq 1/\beta$. Thus, the proof is completed. 
\end{proof}

\end{document}